\documentclass[12pt,authoryear]{elsarticle}

\usepackage[utf8]{inputenc}
\usepackage[T1]{fontenc}
\usepackage{lmodern}
\usepackage[a4paper, margin=1in]{geometry}

\usepackage{amsmath, amssymb, amsfonts, bm}
\usepackage{array, booktabs, makecell, siunitx, xcolor}
\DeclareSIUnit{\bp}{bp}
\usepackage{enumitem}

\usepackage{graphicx}
\usepackage{float}      
\usepackage{placeins}   
\usepackage{caption}
\usepackage{subcaption}
\usepackage{algorithm}
\usepackage{algorithmic}
\usepackage[table]{xcolor}
\usepackage[authoryear]{natbib}
\usepackage{setspace}
\usepackage{tikz}
\usetikzlibrary{arrows.meta,calc}
\usepackage{pgfplots}
\pgfplotsset{compat=1.18}

\usepackage[hidelinks]{hyperref}
\usepackage{xurl}
\usepackage{amsthm}
\newtheorem{proposition}{Proposition}
\newtheorem{lemma}{Lemma}
\newtheorem{remark}{Remark}

\usepackage{algorithm}
\usepackage{algorithmic}

\newcommand{\good}[1]{\cellcolor{green!20}#1}
\newcommand{\bad}[1]{\cellcolor{red!15}#1}

\begin{document}

\begin{frontmatter}

\title{Revisiting the Structure of Trend Premia: When Diversification Hides Redundancy}

\author[1,2]{Alban Etienne}
\author[1]{Jean-Jacques Ohana}
\ead{jean-jacques.ohana@aiforalpha.com}
\author[1,3]{Eric Benhamou}
\author[1]{Béatrice Guez}
\author[1]{Ethan Setrouk}
\author[1]{Thomas Jacquot}

\address[1]{Ai For Alpha, France and USA} 
\address[2]{Centrale Lyon, France} 
\address[3]{Université Paris Dauphine–PSL}

\begin{abstract}
Recent work has emphasized the diversification benefits of combining trend signals across multiple horizons, with the medium-term window—typically six months to one year—long viewed as the “sweet spot” of trend-following. This paper revisits this conventional view by reallocating exposure dynamically across horizons using a Bayesian optimization framework designed to learn the optimal weights assigned to each trend horizon at the asset level. The common practice of equal weighting implicitly assumes that all assets benefit equally from all horizons; we show that this assumption is both theoretically and empirically suboptimal. We first optimize the horizon-level weights at the asset level to maximize the informativeness of trend signals before applying Bayesian graphical models—with sparsity and turnover control—to allocate dynamically across assets. The key finding is that the medium‑term band contributes little incremental performance or diversification once short‑ and long‑term components are included. Removing the 125-day layer improves Sharpe ratios and drawdown efficiency while maintaining benchmark correlation. We then rationalize this outcome through a minimum-variance formulation, showing that the medium-term horizon largely overlaps with its neighboring horizons. The resulting “barbell” structure—combining short- and long-term trends—captures most of the performance while reducing model complexity. This result challenges the common belief that more horizons always improve diversification and suggests that some forms of time-scale diversification may conceal unnecessary redundancy in trend premia.
\end{abstract}

\begin{keyword}
Trend premia \sep managed futures \sep systematic investing \sep time-scale decomposition \sep performance attribution \sep portfolio diversification

\vspace{0.5em}
\textit{JEL classification:} G11 \sep G12 \sep C58
\end{keyword}

\end{frontmatter}

\clearpage 
\tableofcontents

\section{Introduction}\label{sec:Introduction}
\subsection{Motivations}
Trend-following strategies are among the most persistent and well-documented sources of excess returns in financial markets. Managed futures funds, or Commodity Trading Advisors (CTAs), systematically exploit directional persistence in prices across equities, bonds, commodities, and currencies through long–short positions in liquid futures. These strategies have been shown to deliver returns that are largely uncorrelated with traditional risk premia and that tend to perform strongly during equity market drawdowns, a property often described as “crisis alpha” (\citealp{MoskowitzOoiPedersen2012}; \citealp{Hurst2017}). As a result, CTAs have become a cornerstone of institutional portfolios and a central object of study in the literature on alternative risk premia and dynamic asset allocation.

A defining feature of most trend-following systems is the use of multiple lookback horizons to estimate and trade price trends. Combining short-, medium-, and long-term signals is widely viewed as a form of time-scale diversification that improves robustness to different market environments (\citealp{BaltasKosowski2013Momentum}; \citealp{Baz2015DissectingInvestmentStrategies}; \citealp{Lemperiere2017}). The medium-term component—typically corresponding to three to six months—has long been regarded as one of the most effective horizons, balancing responsiveness to new information with resilience to noise. Indeed, both academic research and practitioner studies have historically identified this range as the “sweet spot” of trend-following efficiency (\citealp{Fung2001risk}; \citealp{Winton2013}; \citealp{Hurst2017}). 

However, the conventional practice of assigning equal weights to trend horizons implicitly assumes that all assets benefit equally from each horizon. This assumption overlooks the heterogeneity of asset-specific trend dynamics: certain instruments exhibit stronger persistence at short or long horizons, while others display limited predictability across intermediate scales. As a result, equal weighting across horizons may dilute informative signals and reduce overall efficiency. 

In this paper, we address this limitation by dynamically allocating the weights assigned to each trend horizon through Bayesian optimization. The objective is to identify, at the asset level, the combination of horizons that maximizes the informativeness of the trend signals used in the decoding process. This pre-decoding optimization aims to make the input layer itself more predictive, rather than relying solely on post-decoding model adjustments. The Bayesian framework provides a natural mechanism for balancing exploration and regularization, while the inclusion of sparsity and persistence constraints ensures stability through time and prevents overfitting.

After optimizing the weights dynamically across horizons (20, 60, 125, 250 and 500 days), we find that the medium-term layer (125 days) contributes very little to performance or diversification once short- and long-term components are included. Excluding this layer consistently improves Sharpe ratios and drawdown-adjusted performance while maintaining benchmark correlation. The optimal configuration thus takes a “barbell” form: short-term horizons provide convexity and reactivity, while long-term horizons capture persistent macroeconomic trends. The medium-term component, by contrast, appears redundant—largely explained by adjacent horizons—and can be removed without loss of performance.

These findings challenge the traditional view that adding more horizons necessarily enhances robustness. Instead, our results suggest that excessive layering across similar time scales may conceal structural redundancy, creating the illusion of diversification while introducing unnecessary complexity. By systematically testing this hypothesis through a Bayesian dynamic allocation framework, we provide empirical evidence that true diversification arises from combining distinct, complementary horizons rather than overlapping ones. 
A stylized toy model further illustrates the intuition behind this redundancy: the medium-term component can be interpreted as a linear combination of short- and long-term trends. Consequently, its inclusion contributes little incremental alpha and may even impair performance, as shown in the model assuming a symmetric Toeplitz correlation structure among horizons.

\subsection{Structure of the Paper}

The remainder of the paper is organized as follows.  
Section~\ref{sec:literature} reviews the empirical and theoretical literature on trend premia, emphasizing the historical role of the medium-term horizon and the debate between diversification and redundancy across time scales.  
Section~\ref{sec:theory} develops the theoretical framework for optimal allocation across trend horizons, formalizing the problem as a mean–variance optimization among correlated horizon-specific factors and deriving analytical conditions under which the medium-term horizon becomes redundant, leading to a barbell structure between short and long horizons.  
Section~\ref{sec:methodology} introduces the Bayesian graphical model that links portfolio returns to asset-level trend scores, and presents the decomposition of each trend score across multiple horizons. This section establishes the two-layer structure distinguishing the estimation of time-varying asset exposures from the construction of horizon-level inputs.  
Section~\ref{sec:results} details the methodology and out-of-sample results of the dynamic optimization of horizon weights, which replaces the naive equal-weighting scheme by an adaptive, data-driven allocation.  
Section~\ref{sec:Excluding} examines the empirical performance of a model that excludes the medium-term horizon and equally weights the remaining four horizons, providing a robustness check on the contribution of intermediate time scales.  
Section~\ref{sec:discussion} discusses the broader implications of these results for diversification theory and systematic trend allocation, highlighting how the findings challenge the assumption that including more horizons necessarily enhances robustness.  
Finally, Section~\ref{sec:conclusion} concludes by summarizing the main insights and outlining directions for future research on the dynamic structure of trend premia and cross-horizon dependencies.

\section{Background and Literature Review}\label{sec:literature}

\subsection{Empirical Foundations of Trend Premia}
A large body of work has examined the performance of trend-following strategies across different time horizons. Early studies established that price trends represent a pervasive and persistent feature of financial markets (\citealp{Fung2001risk}; \citealp{Hurst2017}). Within this literature, the medium-term horizon—typically spanning six months to one year—has long been viewed as the most effective range for capturing the trend premium. Empirical evidence supports this view. For instance, \citet{Winton2013} reports that between 1984 and 2013, an “intermediate-speed” strategy with a lookback of several months achieved a Sharpe ratio of 1.12, compared with 0.87 and 0.81 for short- and long-term strategies, respectively. Similarly, \citet{Hurst2017} find that twelve-month momentum strategies deliver strong and stable returns over more than a century of global data, while shorter-term implementations tend to be noisier and less persistent.

However, the relationship between signal horizon and performance is complex and time-varying (\citealp{ShiLian2025}). Several studies document that very short-term trend strategies—typically based on signals of a few days or weeks—have experienced a progressive decline in effectiveness as markets have become more informationally efficient (\citealp{Man2025}; \citealp{Baz2015DissectingInvestmentStrategies}). This decline is consistent with the semi-strong form of market efficiency, which suggests that information is quickly incorporated into prices, leaving limited scope for exploiting short-lived price continuation. Conversely, very long-term signals, extending beyond one year, tend to suffer from lag effects and partial mean reversion, leading to delayed entries and missed reversals (\citealp{MoskowitzOoiPedersen2012}; \citealp{martin2012momentum}; \citealp{Larry_Swedroe_2022};). Taken together, these findings imply that the most extreme horizons—very short or very long—offer limited standalone value, while intermediate horizons historically provided a more balanced risk–return profile.

\subsection{Diversification across time scales}
Beyond the performance of individual horizons, much of the literature emphasizes the benefits of combining multiple lookback windows within a single trend-following system. Time-scale diversification—mixing fast, medium, and slow signals—is widely regarded as a key source of robustness across market regimes (\citealp{BaltasKosowski2013Momentum}; \citealp{Lemperiere2017}; \citealp{Baz2015DissectingInvestmentStrategies}). Empirically, correlations between signals at different horizons are often modest, providing diversification benefits similar to those observed across asset classes. For example, \citet{Man2023} document a correlation of only 0.17 between very fast and very slow trend models within their multi-horizon CTA framework. This low correlation helps smooth performance over time, especially during volatile periods or abrupt regime shifts.

Despite this established consensus, relatively few studies have examined whether the contribution of the medium-term component remains distinct once short- and long-term signals are jointly used. In practice, most large CTAs implement a dense grid of lookback windows to cover a broad spectrum of trend speeds (\citealp{DolfinMaxey2017}). Yet it remains unclear whether each layer provides independent diversification or whether some horizons are redundant—capturing information already embedded in adjacent ones. This open question lies at the intersection of two strands of literature: the study of trend premia as a systematic return source, and the broader portfolio-theoretic discussion of how diversification interacts with redundancy in correlated strategies.

The present paper contributes to this discussion by systematically isolating the marginal value of the medium-term horizon in multi-horizon trend-following strategies. By decomposing trend premia into five distinct lookback components and evaluating the impact of removing the 125-day layer, we provide new evidence on whether time-scale diversification genuinely enhances performance or simply overlays similar exposures across adjacent horizons.

\subsection{Redundancy and Horizon Selection in Trend Systems}
While many trend-following systems combine a dense grid of lookback windows, several studies suggest that most of the diversification benefit can be achieved with only a small number of distinct horizons. \citet{Larry_Swedroe_2022} show that a dynamic allocation switching between short- and long-term signals—corresponding to one-month and twelve-month lookbacks—outperforms both static combinations and any single signal in isolation. Likewise, \citet{GouldingHarveyMazzoleni2023} construct an intermediate indicator based on a time-varying combination of only two horizons, demonstrating that a parsimonious design can capture much of the performance traditionally attributed to a full spectrum of speeds. These results imply that the incremental value of intermediate horizons is limited once the extremes are represented. Complementary evidence from \citet{BenhamouAugmentedST} further shows that short-term signals play an essential role in mitigating drawdowns and capturing early trend reversals, while multi-horizon combinations outperform single-horizon strategies by providing stability across market regimes. Together, these studies highlight that effective time-scale diversification may require far fewer layers than conventionally assumed, raising questions about the true necessity of the medium-term component.

\paragraph{Limitations of the literature and contribution of this study}
Despite a rich literature on trend-following and time-horizon diversification, one critical question remains open: \emph{can horizon exposures be dynamically adjusted over time to improve robustness, rather than relying on static equal weighting?}
Most academic and practitioner research—such as that of AlphaSimplex, AQR, and Man Group—advocates combining multiple horizons, typically using fixed or equal weighting schemes to mitigate overfitting and ensure model stability. However, such static allocations implicitly assume that the relative importance of each horizon is constant across time and assets. This assumption neglects the time-varying nature of market regimes and the evolving persistence of price trends.
\citet{GurnaniHentschel2021} highlight the role of short-term components in preserving convexity, yet the potential redundancy of the medium-term layer—and, more generally, the lack of adaptivity in horizon allocation—has not been explicitly addressed.
To our knowledge, no prior study has implemented a systematic, out-of-sample framework allowing for \emph{dynamic, Bayesian reallocation of weights across horizons} to test whether the medium-term component continues to add independent value once short- and long-term signals are optimally adjusted through time.

\subsection{Contribution of the paper}
This paper addresses this gap by focusing on the construction of asset-level trend scores across multiple horizons and their aggregation into portfolio returns. The results obtained from the dynamic allocation of horizon weights, together with theoretical considerations, indicated that the medium-term horizon contributed little to overall performance. This observation motivated the development of a strategy in which the 125-day component is systematically excluded and the remaining horizons equally weighted, allowing for an assessment of the impact on Sharpe ratio, drawdowns, and diversification. In both approaches, the resulting asset-level scores are then input into a rolling Bayesian graphical model that allocates exposures across assets. The findings demonstrate that equal weighting across horizons is suboptimal and that excluding the medium-term horizon produces a simpler, more efficient two-scale structure that captures the essence of trend-following across markets and regimes.

\section{Theoretical Framework: Optimal Allocation Across Horizons}\label{sec:theory}
\subsection{Trend Horizons as Correlated Factors}

The allocation of weights across trend horizons can be formalized in a mean–variance framework analogous to Markowitz’s portfolio theory. Each horizon specific trend signal can be viewed as a distinct, yet correlated, return factor that captures price persistence over a particular time scale. From this perspective, optimizing horizon weights is equivalent to constructing an efficient portfolio of trend factors—balancing the trade-off between expected return, volatility, and diversification. This formulation provides a theoretical foundation for assessing the marginal contribution of each horizon to overall performance and for determining whether certain horizons, such as the medium-term component, offer unique diversification benefits or merely replicate information embedded in neighboring horizons.

\subsection{Minimum-Variance Allocation and Economic Interpretation}
A classical approach to quantifying the contribution of each trend horizon is to treat single-horizon strategies as individual assets within a minimum-variance portfolio. In this setting, each “asset” corresponds to a trend-following strategy based solely on signals from a given lookback window—20, 60, 125, 250, or 500 trading days. The return of each single-horizon strategy thus represents the performance that would be achieved by following trends at one particular time scale in isolation. 

By estimating the covariance matrix of these horizon-specific returns, the optimization identifies the set of weights that minimizes overall portfolio variance subject to a full-investment constraint. This minimum-variance solution reveals how diversification operates across horizons: horizons that are highly correlated with others or contribute little specific variance receive smaller or zero weights, while those that add independent sources of return volatility receive larger allocations. In this framework, a horizon’s optimal weight reflects its incremental economic value to the portfolio of trend signals, providing a rigorous and interpretable measure of redundancy or distinctiveness among time scales.

\subsection{Minimum-Variance Allocation Across Horizons}
Let $r = (r_1, r_2, \dots, r_H)^\top$ denote the expected returns of the $H$ single-horizon strategies and  $\Sigma$ denote their covariance matrix. We can formulate the standard minimum-variance portfolio as follows

\begin{proposition}[Minimum-variance portfolio with full investment]\label{prop:min_variance_portfolio}
Let $\Sigma \in \mathbb{R}^{H\times H}$ be symmetric positive definite and $\mathbf{1}\in\mathbb{R}^H$ the vector of ones.
The optimization problem
\begin{align}
\min_{w \in \mathbb{R}^H} \quad & w^\top \Sigma w, \\
\text{s.t.} \quad & w^\top \mathbf{1} = 1,
\end{align}
has the unique solution
\begin{equation}
w^\star = \frac{\Sigma^{-1}\mathbf{1}}{\mathbf{1}^\top \Sigma^{-1}\mathbf{1}}.
\end{equation}
\end{proposition}

\begin{proof}
Refer to Proof~\ref{appendix:proof_meanvariance} for technical details, which is a reformulation or direct consequences of classical arguments in mean–variance analysis and convex optimization as presented in ~\cite{Markowitz1952} and ~\cite{BoydVandenberghe2004}.
\end{proof}
This solution defines the combination of horizon-specific trend strategies that minimizes portfolio variance subject to full investment. The relative weights $w^\star$ are proportional to the inverse of each horizon’s contribution to total portfolio risk, adjusted for cross-horizon correlations. In this framework, each horizon is treated as a factor that contributes to overall return variance, allowing its incremental value to be assessed in a unified, covariance-based setting.

\paragraph{Economic Interpretation}
The minimum-variance allocation provides a natural diagnostic for identifying redundant horizons. A horizon receiving a zero or negative optimal weight offers little or no incremental diversification, as its return pattern can be replicated by a linear combination of other horizons. Conversely, horizons with large positive weights capture distinct risk exposures or low correlations with other trend components. This perspective interprets the weight structure not merely as a mathematical outcome but as an economic measure of each horizon’s marginal contribution to the efficiency of the overall trend portfolio.

\paragraph{Economic Intuition for Excluding the Medium-Term Horizon}
The optimization problem can be simplified by aggregating the five original horizons into three representative components—short, medium, and long. Empirically, the medium-term horizon is highly correlated with its neighboring short- and long-term signals. When this correlation approaches the level of near collinearity, the medium-term component provides little incremental diversification. Its risk–return profile can be effectively replicated by a convex combination of the two extremes. Consequently, the optimal allocation tends to concentrate on the short- and long-term horizons, forming a barbell structure across time scales. This pattern reflects an economic rather than a purely statistical outcome: the medium-term horizon is redundant because it captures trends already embedded in the faster and slower components.

\subsection{A Stylized Three-Horizon Model and the Barbell Allocation}

To formalize this intuition, consider three representative trend strategies—short ($T_1$), medium ($T_2$), and long ($T_3$)—each with identical expected excess returns $\mu>0$ and volatility $\sigma>0$. Assume that the correlation matrix among these horizons follows a symmetric Toeplitz form and is given by

\[
R(\rho,\delta) =
\begin{pmatrix}
1 & \rho & \delta \\
\rho & 1 & \rho \\
\delta & \rho & 1
\end{pmatrix}, \qquad \delta \in [0,1), \quad \rho \in (0,1),
\]

where $\rho$ denotes the correlation between adjacent horizons (short-medium or medium-long), and $\delta$ the correlation between the two extremes. Portfolio weights $w = (w_1,w_2,w_3)^\top$ satisfy the non-negativity and budget constraints $w_i \ge 0$ and $w_1+w_2+w_3=1$. 

The covariance matrix of horizon returns is $\Sigma = \sigma^2 R(\rho,\delta)$, and total portfolio variance is
\[
\sigma_p^2 = w^\top \Sigma w = \sigma^2 w^\top R w.
\]
Maximizing the Sharpe ratio, assuming a zero risk-free rate, is equivalent to solving

\[
\min_{w \ge 0,\, 1^\top w = 1} \; w^\top R w.
\]

\begin{proposition}[Exclusion of the Intermediate Horizon] \label{prop:barbell}
Let $R(\rho,\delta)$ be defined as above. Then $R(\rho,\delta) \succ 0$ if and only if $\rho < \sqrt{(1+\delta)/2}$. Moreover, when $\rho \ge (1+\delta)/2$, the unique minimum-variance allocation is the barbell portfolio
\[
w^\star = \left( \tfrac{1}{2},\, 0,\, \tfrac{1}{2} \right),
\]
with minimum portfolio variance
\[
\sigma_p^{\star 2} = \sigma^2 \tfrac{1+\delta}{2}, \quad 
\sigma_p^\star = \sigma \sqrt{\tfrac{1+\delta}{2}}.
\]
\end{proposition}

\begin{proof} See ~\ref{proof:barbell}
\end{proof}

\paragraph{Interpretation}
Proposition~\ref{prop:barbell} demonstrates that when adjacent trend horizons exhibit high pairwise correlation relative to the correlation between the extremes, the intermediate horizon offers no incremental diversification. In such environments, the minimum-variance frontier is spanned by the short and long horizons alone, yielding a barbell allocation across time scales. The result provides a theoretical foundation for our empirical observation that excluding the medium-term horizon leaves performance largely unaffected—and in several cases enhances efficiency by removing redundant exposures. Conceptually, the medium-term horizon functions as a statistical intermediary rather than an independent source of trend premia, capturing patterns already reflected in faster and slower signals.

The assumption of identical expected excess returns ($\mu>0$) and volatilities ($\sigma>0$) across horizons, together with the symmetric Toeplitz structure of the correlation matrix, serves as a normalization device that isolates the contribution of cross-horizon dependence. By holding the marginal distribution of each sleeve constant, the analysis abstracts from differences in signal strength or volatility scaling and focuses exclusively on the role of correlation in shaping optimal portfolio weights. Economically, this specification treats each horizon as an equally efficient transformation of a common underlying trend factor, differing only in timing and interdependence. The symmetry assumption thus enables a clean identification of the diversification mechanism: it reveals that the emergence of a barbell allocation is driven by the structure of horizon correlations rather than by heterogeneity in mean returns or risk. In applied settings, heterogeneity in $\mu$ and $\sigma$ can be reintroduced as a second-order refinement once the structural correlation effect has been characterized.

\paragraph{Practical Implications}
The theoretical insight that redundant horizons are optimally assigned negligible weights has direct implications for the design of multi-horizon trend systems. Weighting schemes that favor a small number of distinct, weakly correlated horizons can achieve higher risk-adjusted returns with lower turnover and estimation error. In contrast, the indiscriminate layering of adjacent horizons—often justified on diversification grounds—can obscure underlying redundancy and dilute overall efficiency. Empirically, we find that horizon allocations concentrated on the short and long ends of the spectrum yield more stable and interpretable exposures, offering a parsimonious representation of the trend premium that remains robust across asset classes and market regimes.

\paragraph{Link to Pre-Decoding Optimization}
In our framework, the allocation across horizons is determined prior to the decoding stage, with the graphical model operating on a weighted aggregation of horizon-specific trend signals. Although the empirical analysis focuses on post-decoding performance, the theoretical implications of the pre-decoding optimization remain directly relevant. 

First, it provides a natural measure of each horizon’s marginal utility in terms of both diversification and contribution to portfolio efficiency. Horizons that are highly correlated with others and exhibit limited specific variance are assigned negligible weights, reflecting their lack of incremental informational content. Second, excluding such redundant horizons before decoding does not compromise replication accuracy. On the contrary, it can enhance stability by reducing the effective dimensionality of the allocation problem and limiting the propagation of estimation noise. 

Taken together, these insights highlight that the economic value of a horizon lies not in its nominal distinctiveness, but in its ability to convey information that is orthogonal to the other components of the trend structure. The optimization process therefore acts as a filtering mechanism, retaining only those horizons that deliver genuine diversification in the return-generating space.
\section{Empirical Methodology}
\label{sec:methodology}

\subsection{Bayesian Graphical Model for Dynamic Horizon Allocation}
\label{sec:empirical_methodology}

Following \citet{ohana2022deepdecoding} and \citet{benhamou2024generativeai}, We model portfolio returns as the result of exposures to asset-specific trend signals. At each date $t$, the portfolio return is written as
\begin{equation*}
r_t = \sum_{i} w_{t,i} \, x_{t,i} + \varepsilon_t,
\end{equation*}
where $w_{t,i}$ denotes the time-varying exposure to asset $i$, $x_{t,i}$ the corresponding trend score, and $\varepsilon_t$ an idiosyncratic error term. The latent weights evolve smoothly through time according to a Gaussian state equation:
\begin{equation*}
w_{t,i} = w_{t-1,i} + \eta_{t,i}, \qquad \eta_{t,i} \sim \mathcal{N}(0, \tau_i^2),
\end{equation*}
which allows the model to capture gradual adjustments in the relative importance of each asset. This specification defines a Bayesian graphical structure in which the observation layer links realized returns to asset-level trend scores, while the latent layer governs the temporal dynamics of the exposures $\{w_{t,i}\}$.

In the current formulation, each asset-level trend score $x_{t,i}$ is itself constructed as an aggregate of horizon-specific signals:
$$
x_{t,i} = \sum_{h \in \mathcal{H}} w_{t,i,h} \, x_{t,i,h}, 
\qquad \sum_{h \in \mathcal{H}} w_{t,i,h} = 1, 
\qquad \mathcal{H} = \{20, 60, 125, 250, 500\}.
$$

Here, $x_{t,i,h}$ represents the trend-following signal of asset $i$ computed at horizon $h$, and $w_{t,i,h}$ is the weight assigned to that horizon within the composite trend score. In the benchmark implementation used within the graphical model, these horizon weights are assumed to be \emph{equally weighted} (i.e., $w_{t,i,h}=1/|\mathcal{H}|$), implying that all horizons contribute identically to $x_{t,i}$.

It is important to distinguish between these two sets of weights:
\begin{itemize}
    \item The graphical model \emph{estimates} the exposures $\{w_{t,i}\}$ that map the trend scores $\{x_{t,i}\}$ into portfolio returns.
    \item The horizon weights $\{w_{t,i,h}\}$, which define how each $x_{t,i}$ is built from horizon-level components $\{x_{t,i,h}\}$, are \emph{inputs} to the graphical model and are fixed \emph{a priori} in the current setting.
\end{itemize}

The objective of our work is precisely to improve upon this naïve equal weighting by determining a set of \emph{optimal, dynamically adjusted} horizon weights $\{w_{t,i,h}\}$. These optimized horizon allocations produce more informative asset-level trend scores $\{x_{t,i}\}$, which then serve as superior inputs to the Bayesian graphical model. In other words, while the graphical model learns the time-varying asset exposures $\{w_{t,i}\}$, our contribution focuses on how to construct each $x_{t,i}$ most effectively by optimally combining its horizon-specific components.

\subsection{Dynamic Allocation Across Trend Horizons}

Equal weighting across horizons implicitly assumes that all assets respond uniformly to trends at different time scales. In practice, this assumption is unrealistic: the persistence and speed of price adjustments vary across assets and market regimes. A uniform allocation therefore overlooks systematic differences in how trends manifest across horizons.

To address this limitation, we estimate asset-specific horizon weights $\{w_{t,i,h}\}$ that optimally combine the horizon-level trend signals $\{x_{t,i,h}\}$. The optimization is conducted prior to the graphical modeling step, ensuring that the resulting aggregate inputs $\{x_{t,i}\}$ capture the most informative mixture of horizons for each asset. Conceptually, this procedure identifies the combination of short-, medium-, and long-term signals that maximizes risk-adjusted performance subject to smoothness and persistence constraints.

Beyond static optimization, we assess the temporal stability of the learned weights. Stable horizon weights indicate that the contribution of each horizon reflects enduring structural or behavioral features of market trends, whereas frequent or abrupt changes suggest regime dependence or transient noise. We evaluate stability over rolling subperiods using three diagnostics: (i) the volatility of the optimized weights, measuring smoothness over time; (ii) their first-order autocorrelation, capturing intertemporal continuity; and (iii) the maximum absolute variation between consecutive periods, reflecting exposure instability. Horizons exhibiting low volatility, high autocorrelation, and limited intertemporal variation are classified as \emph{stable}.

Only when a sufficient proportion of horizons meet these persistence criteria are the estimated weights smoothed using an exponentially weighted moving average. This adaptive smoothing ensures that the final combination of horizons evolves gradually, filtering out transient noise while preserving sensitivity to structural changes. When insufficient stability is detected, the allocation defaults to equal weighting across horizons. 

This two-tiered procedure—combining within-horizon optimization and cross temporal validation—serves to distinguish persistent, information-rich horizons from those whose apparent predictive power is spurious. The resulting trend signals, passed to the decoding stage, are therefore both statistically robust and economically interpretable, capturing trend dynamics that are persistent across time rather than driven by short-lived fluctuations.

\subsection{Controlling for Overfitting Through Persistence Filtering}

A central challenge in empirical asset-pricing models is the risk of overfitting extracting spurious relationships that reflect sample-specific noise rather than persistent, economically meaningful structure. This issue is particularly acute in trend-following strategies, where the underlying return process exhibits limited and time-varying predictability. Blind optimization of horizon weights in such environments can produce allocations that fit historical idiosyncrasies yet fail to generalize out of sample.

To mitigate this risk, we adopt an adaptive framework that conditions the optimization of horizon weights on the persistence of trend signals. The underlying premise is that the informational value of a trend depends not only on its historical performance but also on the temporal stability of its contribution to returns. Assets or markets exhibiting consistent, long-lived trends are more likely to yield reliable predictive content, whereas those with unstable or rapidly reversing patterns are more prone to generating transitory noise.

Operationally, for each asset, we estimate the optimal combination of horizon weights across multiple training subperiods. The temporal evolution of these weights is then assessed through a set of stability diagnostics: (i) the standard deviation of weights across subperiods, (ii) the maximum change between adjacent periods, (iii) first-order autocorrelation, and (iv) monotonicity of directional behavior. These diagnostics jointly capture the smoothness, persistence, and directional coherence of horizon exposures. 

Assets for which the stability indicators exceed a defined persistence threshold are classified as \emph{predictable}. For these assets, the optimized horizon weights are retained for out-of-sample implementation. Conversely, when weight paths display high volatility or low temporal correlation—indicating structural instability—the asset is labeled \emph{non-predictable}, and the allocation defaults to an equal-weighted combination of horizons. 

This persistence-filtering mechanism acts as a safeguard against model overfitting by ensuring that the optimization process privileges signals that are both statistically consistent and economically interpretable. In doing so, it enhances the model’s out-of-sample robustness and prevents the artificial inflation of in-sample performance due to transient, noise-driven relationships.

This framework offers two key advantages. First, it mitigates the risk of overfitting by activating the optimization process only when the underlying trend signal exhibits sufficient persistence and informational strength. This conditional approach limits the likelihood that allocations are driven by sample-specific noise rather than genuine return structure. Second, by restricting the influence of assets with unstable or weakly persistent trends, the model enhances out-of-sample robustness. The resulting allocations are less sensitive to transient fluctuations and more representative of the structural dynamics governing market behavior.

Empirically, the persistence of horizon weights is assessed by partitioning the historical sample into a series of subperiods—for example, sixteen half-year windows within an eight-year training interval. For each subperiod, optimal horizon weights are estimated, and their temporal stability is then evaluated using a set of statistical indicators, including standard deviation, first-order autocorrelation, and maximum inter-window variation. 

When the estimated weights exhibit high stability and cross-temporal consistency, they are retained for subsequent out-of-sample implementation. Conversely, when instability dominates—indicative of a lack of persistent trend structure—the allocation defaults to a conservative benchmark, typically an equal-weighted combination of horizons. 

This persistence-based selection mechanism ensures that optimization efforts concentrate on markets where trend signals convey genuine predictive content. By filtering out noisy or regime-dependent relationships, the approach preserves only those signals with demonstrable informational durability, thereby generating allocations that remain economically meaningful and resilient across changing market conditions.

\section{Empirical Results}\label{sec:results}

\subsection{Universe Description and Cost Framework}\label{sec:cost-summary}
For any backtest, we use three layers of implementation costs (see Table\ref{tab:futures_universe} for contract-level figures):

\begin{itemize}[itemsep=3pt, parsep=0pt, topsep=0pt, partopsep=0pt]
   \item \textbf{Transaction cost (Tx.Cost).} Round-turn execution expense that bundles bid–ask, brokerage, exchange and clearing fees plus a small slippage buffer.
  \item \textbf{Replication (Roll) cost.} Systematic drag when the front-month future is rolled; measured as the 2000–2025 average front-to-next calendar spread.
  \item \textbf{Management fee.} Flat \(( 50\,\,bps) \) per-annum charge on AUM.
\end{itemize}

\begin{table}[H]
\centering
\resizebox{\linewidth}{!}{%
    \begin{tabular}{@{}ll>{\raggedright\arraybackslash}p{10cm}@{}}
    \toprule
    \textbf{Asset class} & \textbf{Costs (Tx, Roll)} & \textbf{Instruments (exchange)} \\
    \midrule
    Commodities  & 2 / 15 bps & GC (COMEX); CL, NG (NYMEX); CO (ICE); HG (COMEX) \\
    Equity       & 2 / 15 bps & ES, NQ (CME); NK (OSE); SX (Eurex); Z (ICE); EM (CME) \\
    Fixed income & 2 / 10 bps & TU, TY (CBOT); SZ, RX (Eurex); G (ICE); JGB (OSE); XM (ASX) \\
    FX           & 2 / 2 bps  & EUR, JPY, GBP, AUD, CAD (CME) \\
    \bottomrule
    \end{tabular}%
}
\caption{Market futures Universe with Transaction / Roll costs in basis points.}
\label{tab:futures_universe}
\end{table}

\subsection{Replication Framework and Horizon-Level Performance}
The empirical design differs from conventional predictive modeling frameworks in which the objective is to forecast future asset returns directly. Here, the focus is on estimating the optimal allocation of trend weights across horizons to replicate a representative CTA benchmark. Consequently, classical notions of statistical generalization are secondary to the assessment of temporal persistence—specifically, the stability and robustness of the estimated weights across different market environments.

To capture persistent features of market trends, the training windows are deliberately long, spanning eight years—approximately one full economic cycle. This horizon balances the need to observe trend behavior across multiple regimes with the requirement for sufficient sample depth to estimate stable cross-horizon relationships. Each training window is further divided into semiannual subperiods, providing both temporal resolution and statistical reliability. 

This semiannual granularity offers two advantages. First, it produces a sufficiently large number of subperiods—more than fifteen within an eight-year window—allowing for a meaningful assessment of weight stability over time. Second, the six-month duration is long enough to encompass multiple market conditions, including upward, downward, and range-bound phases, thereby reducing the risk that optimization results are dominated by transitory short-term noise.

The validation procedure mirrors this temporal structure. Each six-month test window follows the corresponding training window, with the dataset rolling forward by six months at each iteration. This rolling scheme preserves temporal independence between training and test periods while allowing for continuous updating of estimated horizon weights. The sequential design ensures that model evaluation reflects genuine out-of-sample performance rather than re-optimization on overlapping data.

Empirically, the graphical model begins out-of-sample evaluation in 2009, following an initial calibration period that provides sufficient historical depth for parameter stabilization. The effective backtest thus extends from 2013 onward, corresponding to weights derived from overlapping eight-year estimation windows. This rolling, cycle-length methodology provides a disciplined mechanism for assessing the persistence of horizon contributions and for identifying markets in which trend-following behavior exhibits structural stability rather than random, short-lived fluctuations.

\subsection{Rolling Optimization and Stability-Based Horizon Selection}

The following procedure summarizes the dynamic estimation framework used to update and validate horizon weights over time. The algorithm (described in Algorithm~\ref{alg:rolling_weights}) iteratively re-optimizes the trend-horizon allocation on overlapping training windows and tests its robustness on subsequent validation periods. The objective is to retain only those horizon weights that exhibit sufficient temporal persistence—defined through volatility, autocorrelation, and maximum step-size thresholds—before applying them to the out-of-sample test phase.

\noindent
This procedure ensures that rebalancing decisions are taken only when supported by stable, statistically persistent horizon weights. By combining local optimization with cross-period stability filtering, the model systematically suppresses transient, noise-driven signals while maintaining adaptability to evolving market conditions.

\paragraph{Persistence Rule}\label{stability_test}
Three criteria are used to assess the stability of a series of optimized weights over a training window:

\begin{itemize}
  \item \textbf{Standard deviation of weights.}  
  This must remain below the $40^{\text{th}}$ percentile of the empirical distribution 
  of standard deviations observed across all training windows.  
  This criterion filters out series where weights vary too erratically.

  \item \textbf{First-order autocorrelation.}  
  It must exceed the $60^{\text{th}}$ percentile of the distribution 
  of autocorrelations.  
  This criterion selects series with sufficient persistence, 
  where weights evolve in a consistent manner over time.

  \item \textbf{Maximum variation between consecutive weights.}  
  This variation must remain below the $40^{\text{th}}$ percentile of the empirical distribution 
  of absolute changes between two sub-windows.  
  This criterion filters out series where weights experience significant jumps, 
  which are costly in terms of transactions and difficult to execute in practice.
\end{itemize}

For each asset and each training window, five series of weights are generated, 
corresponding to the five horizons (20, 60, 125, 250, and 500 days).  
A series is considered stable if it satisfies at least two of the three criteria outlined above.  
The set of five weight series is considered globally stable if at least two of the five series are stable.

This approach, based on the $40^{\text{th}}$ and $60^{\text{th}}$ percentiles 
of the corresponding empirical distributions, automatically adapts the thresholds to market conditions.  
When markets are calm, the thresholds become stricter, 
while in periods of high volatility, they loosen.  
This helps to filter out noisy signals while maintaining a sufficient number 
of retained windows to ensure robust out-of-sample validation.

\begin{align}
(1)\quad & \mathrm{std}_T(w_i) \;<\; Q_{40}\Big(\{\mathrm{std}_T(w_j)\}_{j=1}^N\Big) \\[8pt]
(2)\quad & \rho_T^{(1)}(w_i) \;>\; Q_{60}\Big(\{\rho_T^{(1)}(w_j)\}_{j=1}^N\Big) \\[8pt]
(3)\quad & \max_t |w_{i,t} - w_{i,t-1}| \;<\; Q_{40}\Big(\{\max_t |w_{j,t} - w_{j,t-1}|\}_{j=1}^N\Big)
\end{align}

\paragraph{Exponential Moving Average (EMA) Use}

For series considered stable, the weights applied during the test phase 
are not simply computed as the arithmetic average of the sub-windows, 
but as an Exponential Moving Average: 
\[
w^{EMA}_t = \alpha\, w_t + (1-\alpha)\, w^{EMA}_{t-1}, \qquad 0<\alpha \le 1.
\]

This method gives more weight to recent observations, 
while smoothing out short-term fluctuations.  
It thus provides a good balance between stability (noise reduction) 
and adaptability (incorporating new information).

\subsection{Out-of-Sample Results without Persistence Rule: CTA Optimized Trend}
The first approach attempts to optimize the weighting of different assets on each training window, without considering the stability characteristics of the optimized weights series. In the following, this strategy will be referred to as \textit{CTA Optimized Trend}. This method is based on the idea that pure optimization of the weights should maximize the performance in-sample, but it ignores the risk that these weights may reflect temporary fluctuations rather than structural trends.

\paragraph{Utility and iso-curves}
We summarize the trade-off between efficiency and benchmark fit with a Cobb--Douglas utility
\[
U_\alpha = \left(\frac{\text{Return}}{\text{MaxDD}}\right)^{\alpha}\cdot \text{Corr}^{\,1-\alpha}, \qquad \alpha\in(0,1),
\]
and plot iso-curves for $\alpha=0.8$. We report $U_{0.8}$ alongside the individual components.

\begin{figure}[H]
\centering
\includegraphics[width=0.7\textwidth]{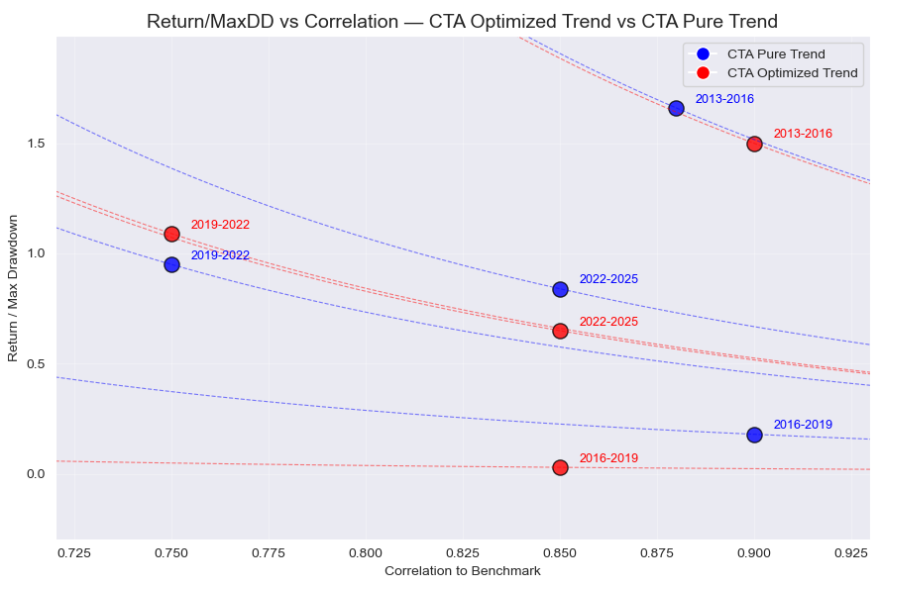}
\caption{Comparison of CTA Optimized Trend and CTA Pure Trend Strategies ($\alpha = 0.8$).}
\label{fig:Comparison_Optimized_vs_Pure_Trend}
\end{figure}

On these three-year rolling windows, and as shown in Figure~\ref{fig:Comparison_Optimized_vs_Pure_Trend}, the optimized allocation strategy (\textit{CTA Optimized Trend}, in red) consistently underperforms the equal-weight benchmark (\textit{CTA Pure Trend}, in blue). Each point in the figure represents a distinct subperiod, spanning from 2013--2016 to 2022--2025, and plots the Return/MaxDD ratio (vertical axis) against the correlation to the benchmark (horizontal axis). The iso-utility curves, parameterized by $\alpha = 0.8$, trace indifference levels between return and correlation, allowing a visual comparison of the two strategies' efficiency in risk-adjusted terms.

A detailed inspection of the data points confirms our finding. During the earliest window (2013–2016), the CTA Pure Trend achieved a Calmar ratio —computed as the Return/MaxDD ratio— close to 1.6 at a benchmark correlation of approximately 90\%, whereas the optimized variant reached only about 1.4 at a similar correlation level.
The next interval, 2016–2019, shows an even sharper divergence: the blue point (Pure Trend) remains around a correlation of 90 \% with a Return/MaxDD ratio slightly above one, while the red point (Optimized Trend) drops below one-half, indicating both lower absolute returns and a higher drawdown profile.
In the 2019–2022 period, when correlation fell to roughly 75 \%, both strategies experienced weaker risk-adjusted outcomes; yet the optimized version still lagged, with a Return/MaxDD ratio close to one, compared with the equal-weight strategy’s slightly superior performance.

Finally, in the most recent 2022–2025 window, both series reconverge around a correlation of 85 \%, but the optimized allocation again remains below its equal-weight counterpart by roughly twenty to thirty basis points in the Return/MaxDD ratio.
These systematic gaps across all four subperiods underline the fundamental instability of the optimized allocation. While the optimizer aims to minimize portfolio variance through horizon-specific weights, the realized performance suggests that these weights overfit to transient in-sample relationships among trend horizons. This results in weights that are ex post suboptimal and ex ante fragile. In contrast, the equal-weight (\textit{CTA Pure Trend}) portfolio, by treating each horizon symmetrically, implicitly regularizes estimation error and mitigates the sensitivity to short-lived covariance shifts.

Economically, the figure demonstrates that the marginal improvement in benchmark correlation obtained by optimization does not translate into superior downside-adjusted performance. The nearly parallel downward-sloping iso-utility curves make this visual: the red trajectory of the optimized strategy remains consistently below the blue one, showing lower utility levels for all observed subperiods. This persistent dominance of the equal-weight allocation aligns with the broader empirical literature emphasizing the instability of mean–variance optimization under parameter uncertainty (see \cite{Michaud1989}; \cite{DeMiguelGarlappiUppal2009}). Taken together, these results highlight that in multi-horizon trend aggregation, structural robustness—rather than statistical precision—drives superior long-term outcomes.

This highlights an important limitation of unconstrained optimization: some market trends may be unstable or changing, and an entirely flexible allocation on each window can produce fragile signals, leading to poor performance when applied to new data.

\subsection{Out-of-Sample Results with Persistence Rule: CTA Dynamic Trend}

\paragraph{Comparison with the CTA Pure Trend Strategy}

\begin{figure}[H]
\centering
\includegraphics[width=0.7\textwidth]{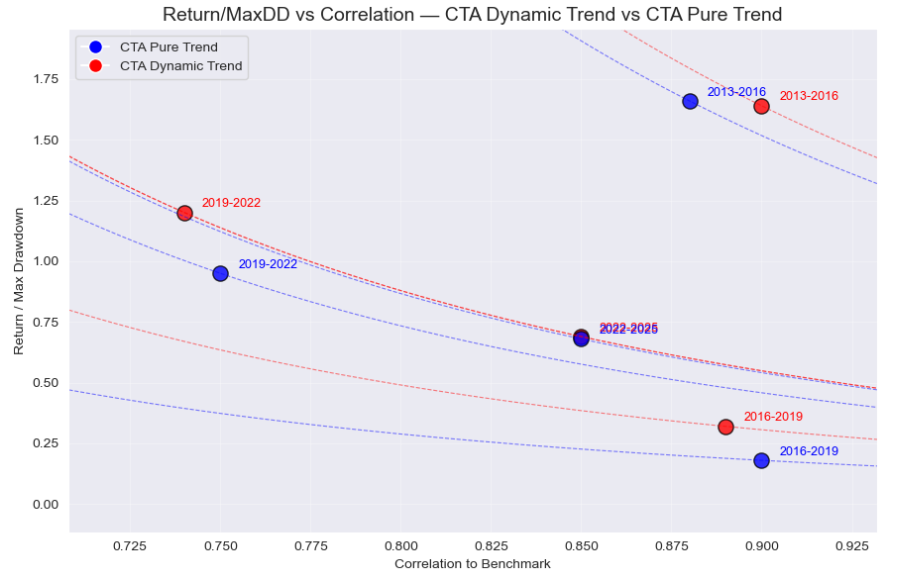}
\caption{Comparison of CTA Dynamic Trend and CTA Pure Trend Strategies ($\alpha = 0.8$).}
\label{fig:Dynamic_vs_Pure_iso}
\end{figure}

\begin{figure}[H]
\centering
\includegraphics[width=0.8\textwidth]{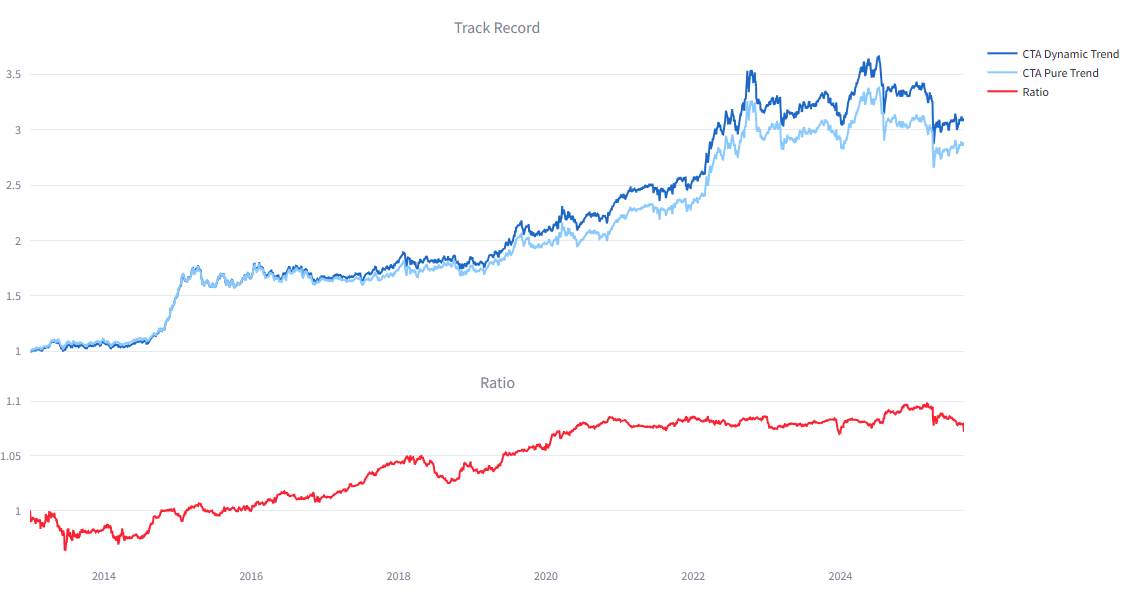}
\caption{Cumulative Performance of CTA Dynamic Trend vs. CTA Pure Trend. The bottom panel shows the performance ratio (Dynamic/Pure).}
\label{fig:Dynamic_vs_Pure_trackrecord}
\end{figure}

As illustrated in Figures~\ref{fig:Dynamic_vs_Pure_iso} and~\ref{fig:Dynamic_vs_Pure_trackrecord}, introducing a persistence rule in the weight allocation (\textit{CTA Dynamic Trend}) materially improves performance stability relative to both the equal-weight (\textit{CTA Pure Trend}) and the fully optimized allocation. The persistence rule enforces gradual weight adjustments across horizons, preventing abrupt reallocations driven by transient covariance shifts. This stabilization mechanism allows the strategy to retain adaptability to changing market regimes while mitigating the effects of overfitting.

Figure~\ref{fig:Dynamic_vs_Pure_iso} summarizes this improvement using three-year rolling windows, plotting the trade-off between Return/Max Drawdown (vertical axis) and correlation to the benchmark (horizontal axis). In each subperiod, the red points (Dynamic Trend) lie on or slightly above the blue ones (Pure Trend), indicating superior or comparable risk-adjusted returns for similar benchmark correlations. The iso-utility contours, corresponding to $\alpha = 0.8$, show that the Dynamic Trend dominates the Pure Trend for all but one period.
A closer examination of the subperiod data provides quantitative support for this dominance. In the early 2013–2016 window, both strategies achieve high efficiency, with Return/MaxDD ratios of 1.66 for the Pure Trend and 1.64 for the Dynamic Trend (Table~\ref{tab:cta-rmdd}). Their benchmark correlations are nearly identical (88 \% versus 90 \%; Table~\ref{tab:cta-corr}). This parity suggests that, in a strongly trending environment, equal weighting performs almost as well as dynamically smoothed allocations.

However, performance diverges sharply during 2016–2019—a period characterized by trend reversals and weaker persistence across asset classes. The Dynamic Trend exhibits a 78\% improvement in the Return/MaxDD ratio (0.32 versus 0.18 for the Pure Trend) and a markedly higher Sharpe ratio (0.24 versus 0.09; Table~\ref{tab:cta-sharpe}), despite a similar benchmark correlation of about 89 \%. This demonstrates the resilience of the persistence rule under low-trend conditions, when frequent reoptimization typically destroys value.

In the subsequent 2019--2022 window, the Dynamic Trend continues to outperform, achieving a Return/MaxDD ratio of 1.20 compared with 0.95 for the Pure Trend. As shown in Figure~\ref{fig:Dynamic_vs_Pure_iso}, both points lie around a benchmark correlation of 0.75, but the Dynamic Trend enjoys a superior position along the utility frontier—delivering roughly 25\% higher efficiency for comparable exposure. The improvement persists into 2022--2025, where Return/MaxDD ratio rises modestly from 0.68 (Pure Trend) to 0.69 (Dynamic Trend), even though overall returns and market trends were weaker during that period. Across the full 2013--2025 horizon, the Dynamic Trend achieves the highest cumulative Sharpe ratio (0.86) and Return/MaxDD ratio (0.74), exceeding both the Pure Trend (0.79 and 0.69, respectively) and the Optimized Trend (0.72 and 0.53).

The cumulative performance plot in Figure~\ref{fig:Dynamic_vs_Pure_trackrecord} corroborates these findings. The top panel shows that the CTA Dynamic Trend (dark blue line) consistently outpaces the CTA Pure Trend (light blue line) over the full sample. The bottom panel plots the cumulative ratio of the two strategies, which remains above 1.0 throughout most of the sample and trends upward steadily after 2016. This indicates sustained outperformance by the dynamic weighting scheme. The ratio peaks around 1.08 in 2023, reflecting an 8\% cumulative gain in performance over the equal-weighted benchmark.

Taken together, the graphical and tabular evidence confirms that persistence-weighted allocation is an effective antidote to overfitting. By introducing inertia in horizon rebalancing, the Dynamic Trend strategy avoids the instability of the fully optimized allocation while capturing more of the long-term convexity inherent in trend-following returns. The quantitative gains reported in Tables~\ref{tab:cta-sharpe}--\ref{tab:cta-corr} demonstrate that this approach delivers superior Sharpe and Return/MaxDD ratios without materially increasing benchmark correlation—an outcome that highlights its efficiency. From a portfolio-construction standpoint, these findings reinforce a broader lesson familiar to the literature on robust optimization (see \cite{Michaud1989}; \cite{DeMiguelGarlappiUppal2009}): statistical simplicity combined with temporal persistence often dominates complex, highly parameterized allocation rules when faced with the realities of noisy and nonstationary financial data.

\begin{table}[H]
\centering
\caption{Sharpe Ratio by Period}
\label{tab:cta-sharpe}
\small
\setlength{\tabcolsep}{2pt}
\begin{tabular}{lccc}
\toprule
Period & Pure Trend & Optimized Trend & Dynamic Trend \\
\midrule
2013--2016 & 1.71 & 1.55 & 1.66 \\
2016--2019 & 0.09 & -0.06 & 0.24 \\
2019--2022 & 0.89 & 0.92 & 1.06 \\
2022--2025 & 0.47 & 0.46 & 0.49 \\
\midrule
2013--2025 & 0.79 & 0.72 & \good{0.86} \\
\bottomrule
\end{tabular}
\end{table}

\begin{table}[H]
\centering
\caption{Return / MaxDD by Period}
\label{tab:cta-rmdd}
\small
\setlength{\tabcolsep}{2pt}
\begin{tabular}{lccc}
\toprule
Period & Pure Trend & Optimized Trend & Dynamic Trend \\
\midrule
2013--2016 & 1.66 & 1.50 & 1.64 \\
2016--2019 & 0.18 & 0.03 & 0.32 \\
2019--2022 & 0.95 & 1.09 & 1.20 \\
2022--2025 & 0.68 & 0.65 & 0.69 \\
\midrule
2013--2025 & 0.69 & 0.53 & \good{0.74} \\
\bottomrule
\end{tabular}
\end{table}

\begin{table}[H]
\centering
\caption{Correlation with Benchmark by Period}
\label{tab:cta-corr}
\small
\setlength{\tabcolsep}{2pt}
\begin{tabular}{lccc}
\toprule
Period & Pure Trend & Optimized Trend & Dynamic Trend \\
\midrule
2013--2016 & 0.88 & 0.90 & 0.90 \\
2016--2019 & 0.90 & 0.85 & 0.89 \\
2019--2022 & 0.75 & 0.75 & 0.74 \\
2022--2025 & 0.85 & 0.84 & 0.85 \\
\midrule
2013--2025 & 0.82 & 0.82 & \good{0.83} \\
\bottomrule
\end{tabular}
\end{table}

\paragraph{Key Takeways}
The results confirm the benefits of dynamically allocating weights across different trend horizons. The chart above shows that, over each three-year sub-period, the \textit{CTA Dynamic Trend} strategy consistently outperforms the equal-weight \textit{CTA Pure Trend} strategy, which does not adjust the weight distribution.

The evolution of the ratio between the two strategies, plotted in red, shows a generally increasing profile above 1. This indicates that adapting the weights to market conditions leads to significant outperformance compared to a static strategy.

The detailed performance tables confirm this trend: the average Sharpe ratio, the Return/MaxDD ratio, and the correlation with the benchmark all highlight an improvement in out-of-sample performance while maintaining consistent correlation with the benchmark. Compared to the optimized version without a persistence rule (\textit{CTA No Medium-Term}), the \textit{CTA Dynamic Trend} strategy demonstrates the importance of reducing overfitting and prioritizing stable weights for persistent horizons.

These observations justify the methodological choice of introducing a persistence rule in the weight optimization process, and suggest that the dynamic allocation strategy is more effective than the purely equal-weighted approach for exploiting market trends across different horizons.

\subsection{Evolution of Optimized Weights}

The weight optimization strategy has consistently underweighted the medium-term horizon across the entire time period. The allocation of weights across the five horizons was as follows: 20\%, 19.5\%, 17\%, 19\%, and 24.5\%, from the shortest to the longest horizon. It is evident that the weights are predominantly concentrated at the extremes, with the shortest and longest horizons receiving a higher share of the total allocation.

This pattern suggests that the medium-term horizon (125 days) is already well-explained by the adjacent horizons (60 and 250 days). As a result, the medium-term horizon appears redundant in the optimization process, contributing little to the overall portfolio performance. Given this observation, it is natural to consider a strategy that excludes the medium-term horizon altogether. By removing the medium-term from the optimization process, we can focus on the more stable, extreme horizons, potentially improving the model’s robustness and performance.

Furthermore, this underweighting of the medium-term horizon is not consistent across all assets. For some assets, the allocation to the medium-term horizon is very low, falling below 10\%, while for others, the allocation is higher, though it never exceeds 22\%. This variability in the medium-term allocation highlights that the exclusion of this horizon might be beneficial in certain cases, where the contribution of the medium-term horizon is minimal or redundant.

\section{Excluding the Medium-Term Horizon}\label{sec:Excluding}

\begin{figure}[H]
\centering
\includegraphics[width=0.85\textwidth]{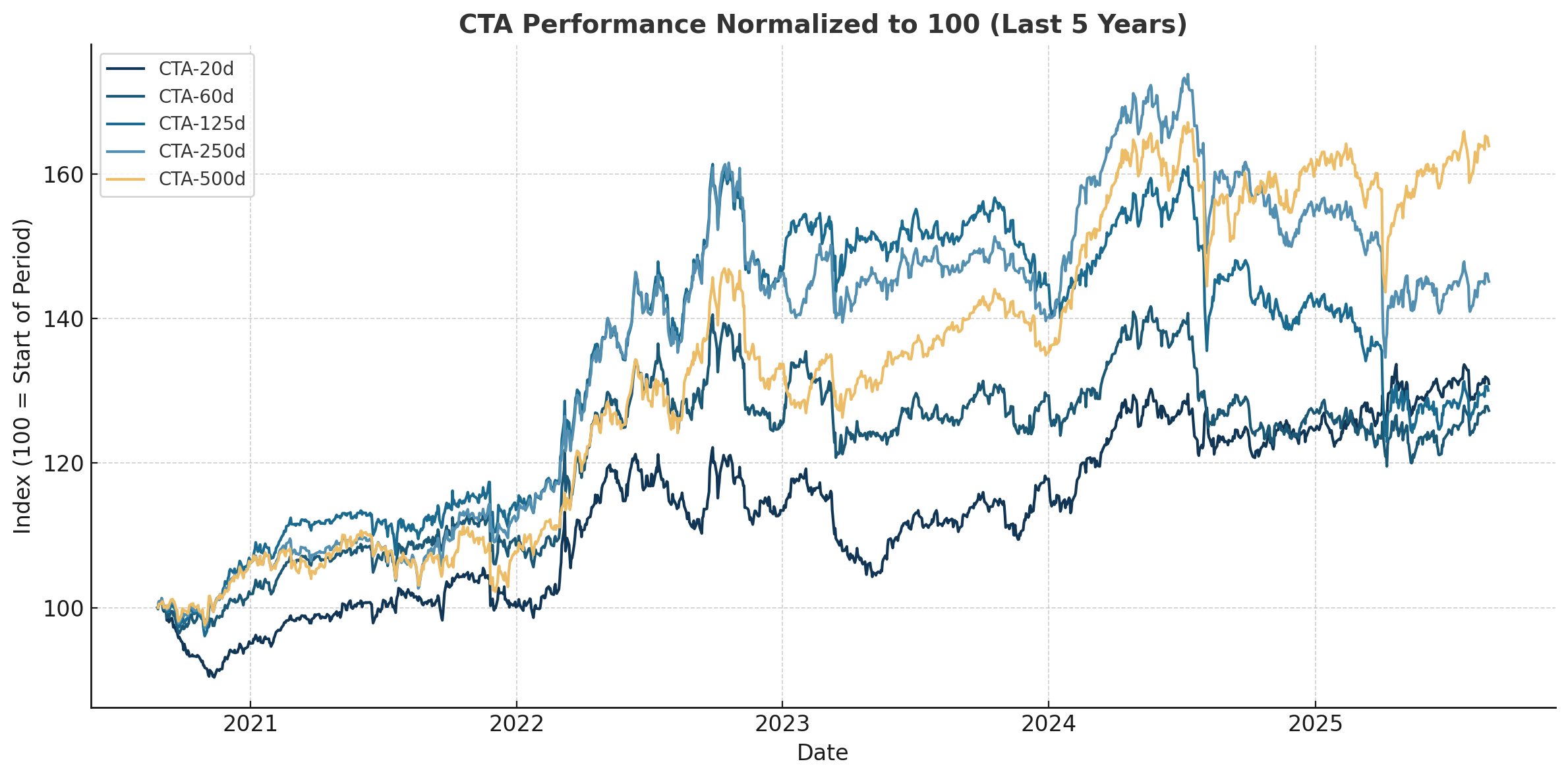}
\caption{Strategy replication based on single-horizon CTA sleeves (20d, 60d, 125d, 250d, and 500d), 2020--2025. Index levels are normalized to 100 at the start of the period.}
\label{fig:trackrecord}
\end{figure}
\paragraph*{Read-through} 
Figure~\ref{fig:trackrecord} provides a comparative visualization of the performance of single-horizon trend-following sleeves, each normalized to an index level of 100 at the start of 2020. The five horizons—20-day, 60-day, 125-day, 250-day, and 500-day—represent progressively slower reactivity to market trends, spanning from short-term tactical trading to long-term directional capture. 

A clear pattern emerges: the 20-day and 500-day sleeves dominate the five-year window, while the 60-day and 125-day horizons persistently underperform. By late 2024, the cumulative index values for the 20d and 500d sleeves reach approximately 160, compared with only about 120 for the 60d and 125d sleeves and around 140 for the 250d horizon. This divergence widens notably after 2022, when macroeconomic volatility and cross-asset dispersion increased, highlighting the structural weakness of medium-term trend signals. 

The 20-day sleeve acts as a diversification anchor: it absorbs market noise and limits drawdowns during reversals, particularly evident in the sharp corrections of early 2022 and late 2023. Its convexity—rapid responsiveness to momentum shocks—preserves portfolio convexity when slower sleeves temporarily decouple from price dynamics. Conversely, the 500-day horizon exhibits strong endurance, steadily compounding through trend persistence in rates, commodities, and FX markets. Its smoother trajectory and smaller sensitivity to short-lived mean reversions enable it to maintain performance even during the choppy markets of 2024–2025.

By contrast, the 60-day and 125-day sleeves—representing medium-term horizons—show repeated episodes of whipsawing. Their trajectories flatten after mid-2022, with limited recovery following market drawdowns. This lagging behavior is consistent with the empirical evidence in Tables~\ref{tab:cta-sharpe}--\ref{tab:cta-rmdd}, where medium-term components contribute disproportionately to risk without adding commensurate returns. Economically, these horizons sit at the intersection between short-term noise and long-term structural trends, resulting in unstable exposure profiles that fail to capitalize on either. Their limited persistence and lack of convexity make them natural candidates for exclusion in multi-horizon aggregation.

The strong post-2023 rebound of the 20d and 500d sleeves—visible in Figure~\ref{fig:trackrecord} as the two uppermost trajectories—confirms that the most effective CTA exposure arises from the combination of the two extremes: the short-term sleeve for reactivity and crisis convexity, and the long-term sleeve for trend persistence and carry capture. Excluding the medium-term components not only enhances overall efficiency but also simplifies the allocation architecture, allowing the dynamic weighting mechanism (as discussed in Section~\ref{fig:Dynamic_vs_Pure_trackrecord}) to focus on structurally differentiated sources of alpha rather than statistically fragile ones.

Taken together, figure~\ref{fig:trackrecord} provides a clear empirical justification for omitting the medium-term horizon. The declining slope of the 60d and 125d trajectories, coupled with their repeated failure to recover after drawdowns, reveals that medium-term signals contribute to noise rather than to risk-adjusted return. This finding aligns with the theoretical predictions from time-series momentum research (\cite{MoskowitzOoiPedersen2012}; \cite{GouldingHarveyMazzoleni2023}), which emphasize that the strongest risk-adjusted returns in trend following arise from either very short or very long lookback windows, not from intermediate ones.

\subsection{Benefit of Multi-Horizons for CTA Replication (2015-2025)}
Over the full decade, the evidence in Tables~\ref{tab:heatmap_full} and~\ref{tab:part1-summary} confirms that performance across horizons is highly heterogeneous, with the long-term and short-term sleeves driving most of the value creation. 

The 500-day strategy emerges as the most efficient single horizon, posting the highest Sharpe ratio (0.47) and the strongest Return/Max DD ratio (0.49) in Table~\ref{tab:part1-summary}. Its annualized return of 7.2\% exceeds that of all other sleeves, while maintaining one of the lowest maximum drawdowns (14.5\%). This efficiency reflects its ability to capture persistent macro trends—most notably the extended directional moves in rates and commodities after 2020—while avoiding excessive turnover. The 250-day horizon performs closely behind, with a Sharpe ratio of 0.42 and a Return/MaxDD ratio of 0.30, reinforcing the notion that slow-moving trend signals provide structural convexity and resilience during volatile macro cycles.

At the opposite end of the spectrum, the 20-day sleeve delivers a modest Sharpe ratio of 0.20 and a return over maximum drawdown ratio of 0.24. However, as shown in Table~\ref{tab:heatmap_full}, its correlation to the benchmark NEIXCTAT (62\%) is the lowest among the horizons, underscoring its role as a diversification engine rather than a benchmark proxy. Its rapid response to short-term market dislocations allows it to preserve convexity during trend reversals—an effect particularly visible in 2022 and early 2024 when long-term sleeves experienced drawdowns. In that sense, the short-term component contributes disproportionately to portfolio asymmetry, despite its standalone efficiency being lower.

The weakest results appear in the 125-day sleeve, which underperforms on all key metrics. With a Sharpe ratio of only 0.21 and a return over maximum drawdown ratio of 0.19 (Table~\ref{tab:part1-summary}), it lags substantially behind both shorter and longer horizons. This horizon’s underperformance is consistent with the “medium-term decay” observed in Figure~\ref{fig:trackrecord}: its signals are too slow to benefit from short-lived momentum bursts yet too fast to capture multi-month macro trends. The 60-day sleeve exhibits a similar, albeit less pronounced, inefficiency, with comparable volatility (10.5\%) but only a slightly better return over maximum drawdown ratio (0.28). Together, these findings reinforce that the medium-term range (60–125 days) adds limited informational value and acts as a drag when integrated into multi-horizon aggregation.

The correlation structure reported in Table~\ref{tab:heatmap_full} provides a complementary perspective. Cross-horizon correlations are highest between the 125d, 250d, and 500d sleeves (0.84--0.90), illustrating the shared exposure of medium- and long-term signals to the same slow-moving market factors. By contrast, the 20d sleeve exhibits only moderate correlation with these longer horizons (44–66\%), highlighting its distinct source of convexity and diversification. Importantly, the All Horizons composite achieves the strongest overall alignment with the CTA benchmark, with a correlation of 0.84 to NEIXCTAT—exceeding any individual sleeve. Among single-horizon models, the 250-day strategy is the closest approximation, with a correlation of 0.81.

From an economic standpoint, this pattern suggests that the benchmark’s behavior reflects a blend of persistent long-term positioning and short-term convexity shocks, rather than medium-term exposure. The high pairwise correlations between long-term sleeves explain their collective stability, while the weak association of short-term horizons with NEIXCTAT supports their complementary diversification role. The \emph{All Horizons} allocation thus benefits from balancing the convexity of short-term signals with the structural persistence of long-term trends—an equilibrium that neither end of the spectrum can achieve in isolation.  

Overall, the decade-long analysis underscores a clear asymmetry in trend efficiency across time horizons. The 500-day sleeve delivers the most stable compounding of returns; the 20-day sleeve contributes meaningful crisis convexity and diversification; and the medium-term sleeves, by contrast, generate limited incremental alpha while amplifying noise. These findings rationalize the exclusion of the medium-term horizon in the dynamic allocation framework presented in Section~\ref{sec:Excluding}, and they provide quantitative evidence that CTA trend replication is best achieved through a bimodal exposure structure—anchored at the short and long ends of the trend spectrum.

\begin{table}[H]
\centering
\caption{Correlation matrix of horizons and benchmarks (2015--2025).}
\label{tab:heatmap_full}
\scriptsize
\setlength{\tabcolsep}{2pt}
\renewcommand{\arraystretch}{1.2}
\begin{tabular}{lccccccc}
\toprule
 & 20d & 60d & 125d & 250d & 500d & All Horizons & NEIXCTAT \\
\midrule
20d          & \cellcolor{blue!30}\textbf{100\%} & \cellcolor{blue!20}83\% & \cellcolor{blue!20}59\% & \cellcolor{blue!8}59\% & \cellcolor{blue!6}44\% & \cellcolor{blue!10}66\% & \cellcolor{blue!10}62\% \\
60d          & \cellcolor{blue!20}83\% & \cellcolor{blue!30}\textbf{100\%} & \cellcolor{blue!15}81\% & \cellcolor{blue!10}60\% & \cellcolor{blue!6}44\% & \cellcolor{blue!15}81\% & \cellcolor{blue!10}69\% \\
125d         & \cellcolor{blue!8}59\%  & \cellcolor{blue!15}81\% & \cellcolor{blue!30}\textbf{100\%} & \cellcolor{blue!20}84\% & \cellcolor{blue!10}67\% & \cellcolor{blue!25}94\% & \cellcolor{blue!15}78\% \\
250d         & \cellcolor{blue!6}44\%  & \cellcolor{blue!10}60\% & \cellcolor{blue!20}84\% & \cellcolor{blue!30}\textbf{100\%} & \cellcolor{blue!25}90\% & \cellcolor{blue!25}90\% & \cellcolor{blue!15}81\% \\
500d         & \cellcolor{blue!4}35\%  & \cellcolor{blue!6}44\%  & \cellcolor{blue!10}67\% & \cellcolor{blue!25}90\% & \cellcolor{blue!30}\textbf{100\%} & \cellcolor{blue!10}78\% & \cellcolor{blue!10}75\% \\
All Horizons & \cellcolor{blue!10}66\% & \cellcolor{blue!15}81\% & \cellcolor{blue!25}94\% & \cellcolor{blue!25}90\% & \cellcolor{blue!10}78\% & \cellcolor{blue!30}\textbf{100\%} & \cellcolor{blue!15}84\% \\
NEIXCTAT     & \cellcolor{blue!10}62\% & \cellcolor{blue!10}69\% & \cellcolor{blue!15}78\% & \cellcolor{blue!15}81\% & \cellcolor{blue!10}75\% & \cellcolor{blue!15}84\% & \cellcolor{blue!30}\textbf{100\%} \\
\bottomrule
\end{tabular}
\end{table}

\begin{table}[H]
\centering
\caption{Performance summary by horizon (2015-08-31 to 2025-08-29).}
\label{tab:part1-summary}
\scriptsize
\setlength{\tabcolsep}{2pt}
\begin{tabular}{lccccccc}
\toprule
 & 20d & 60d & 125d & 250d & 500d & All Horizons & NEIXCTAT \\
Annual Return     & 4.2\%  & 4.4\%  & 4.5\%  & 6.7\%  & 7.2\%  & 6.1\%  & 2.7\%  \\
Vol               & 10.0\% & 10.5\% & 10.9\% & 10.8\% & 10.5\% & 10.6\% & 11.1\% \\
Sharpe Ratio      & 0.20   & 0.21   & 0.21   & 0.42   & \good{0.47} & 0.36   & 0.05   \\
Max DD            & 17.3\% & 15.6\% & 23.8\% & 22.6\% & 14.5\% & 21.6\% & 22.4\% \\
Return over maximum drawdown ratio      & 0.24   & 0.28   & 0.19   & 0.30   & \good{0.49} & 0.28   & 0.12   \\
\bottomrule
\end{tabular}
\end{table}

Furthermore, the hierarchical clustering of horizon strategies, displayed in Figure~\ref{fig:dendrogram}, provides a visual synthesis of the structural relationships between the five CTA sleeves. The dendrogram is based on pairwise distances defined as $1 - \text{correlation}$ over the 2015--2025 period, thereby translating comovement patterns into a geometric representation of similarity.

\begin{figure}[H]
\centering
\includegraphics[width=0.5\textwidth]{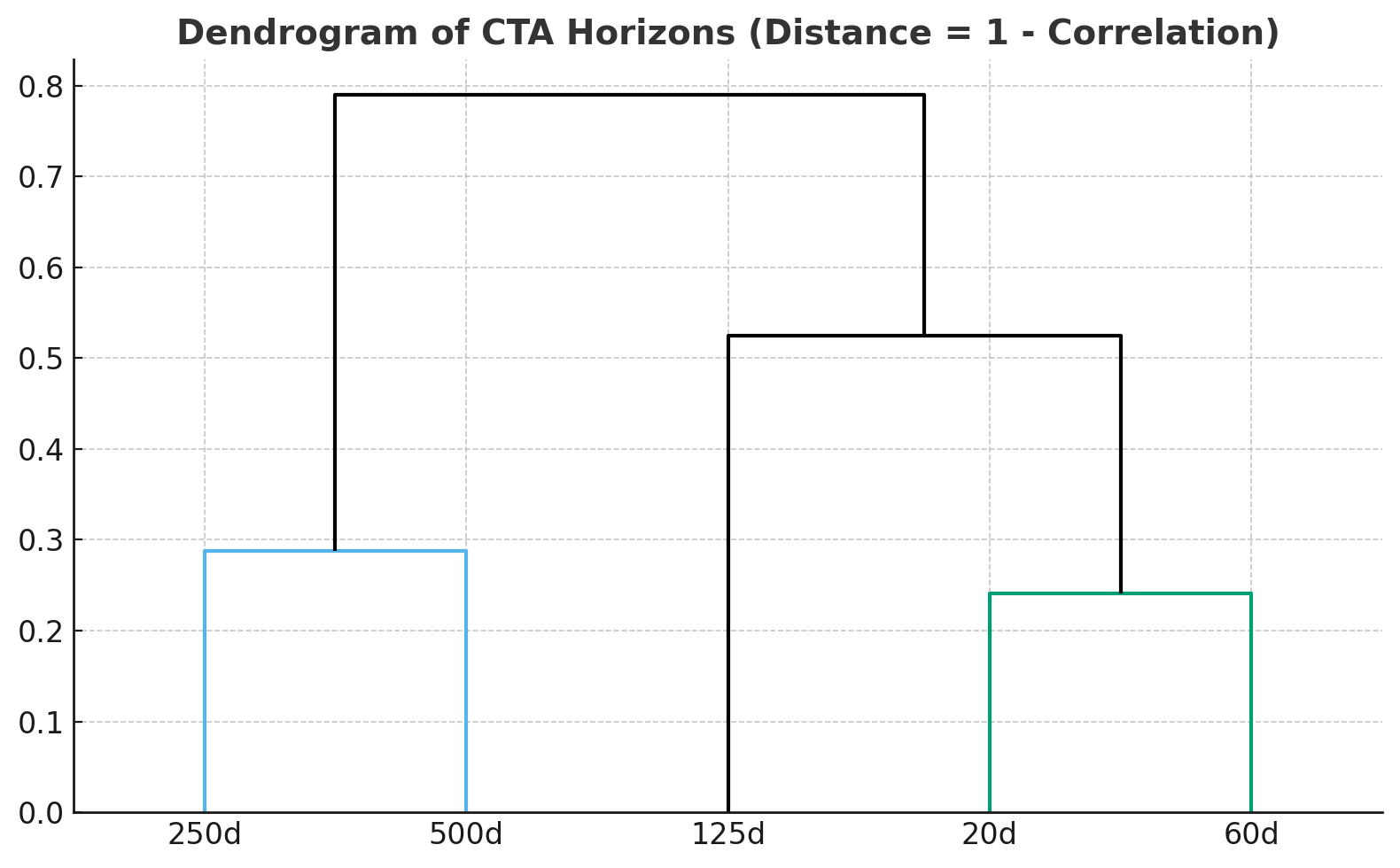}
\caption{Dendrogram of horizon strategies (2015--2025).}
\label{fig:dendrogram}
\end{figure}

As shown in Figure~\ref{fig:dendrogram}, two distinct clusters emerge clearly. The first cluster groups the short-term horizons (20d and 60d), which exhibit a strong degree of co-movement and form a compact subtree with a linkage distance below 0.25, corresponding to pairwise correlations exceeding 0.75 in Table~\ref{tab:heatmap_full}. The second cluster, composed of the long-term horizons (250d and 500d), displays an even tighter relationship, with a linkage distance of approximately 0.30—consistent with their 90\% correlation in the matrix. These two clusters represent economically meaningful dimensions of trend following: short-term reactivity and long-term persistence.

The 125-day sleeve, however, occupies an ambiguous position between the two clusters, joining the dendrogram at a high linkage distance of about 0.55. This position reflects its hybrid exposure profile: it is partially correlated with both short- and long-term signals but insufficiently aligned with either. In other words, the 125d horizon acts as an “intermediate blend” that mirrors characteristics of both extremes without delivering the diversification benefits of the short-term group or the compounding stability of the long-term cluster. 

This structural redundancy explains the empirical inefficiency highlighted earlier in Table~\ref{tab:part1-summary}. The 125d sleeve not only underperforms in terms of Sharpe (0.21) and return over maximum drawdown ratio (0.19) but also fails to diversify the overall portfolio meaningfully—its high correlations with adjacent horizons (0.81 with 60d and 0.84 with 250d) leave limited room for risk reduction. By contrast, the separation between the 20–60d and 250–500d clusters demonstrates that combining the two ends of the trend spectrum yields the most orthogonal, and hence most efficient, exposure structure. 

Economically, this clustering confirms that trend-following efficiency is inherently bimodal. Short-term signals derive value from convexity and responsiveness, while long-term signals capitalize on structural persistence and macro carry effects. The medium-term horizon (125d) contributes neither: it dilutes both effects without improving diversification. The dendrogram in Figure~\ref{fig:dendrogram} thus provides a compact visual argument for excluding the medium-term component from the CTA allocation architecture, reinforcing the empirical evidence presented in Section~\ref{sec:Excluding}.

\subsection{Ablation of Time Horizons: Which Horizons Hurt, Which Add Value?}
To evaluate the marginal contribution of each horizon, we re-run the replication exercise after removing one sleeve at a time and measure performance over disjoint five-year windows. Three complementary objectives are analyzed: (i) the Sharpe ratio, (ii) the  ratio, and (iii) the correlation to the CTA benchmark NEIXCTAT. Table~\ref{tab:ablation-z} summarizes the overall results in standardized (Z-score) form, followed by the detailed period-by-period metrics in Tables~\ref{tab:ablate-sharpe},~\ref{tab:ablate-rmdd}, and~\ref{tab:ablate-corr}.
To evaluate the marginal contribution of each horizon, we re-run the replication exercise after removing one sleeve at a time and measure performance over disjoint five-year windows. Three complementary objectives are analyzed: (i) the Sharpe ratio, (ii) the Return/Max Dratio, and (iii) the correlation to the CTA benchmark NEIXCTAT. Table~\ref{tab:ablation-z} summarizes the overall results in standardized (Z-score) form, followed by the detailed period-by-period metrics in Tables~\ref{tab:ablate-sharpe},~\ref{tab:ablate-rmdd}, and~\ref{tab:ablate-corr}.

\paragraph{Overall ranking}
As shown in Table~\ref{tab:ablation-z}, excluding the 125-day horizon (\emph{No 125}) yields the highest overall Z-score (+0.80), with uniformly positive contributions across all three metrics: +0.84 for Sharpe, +0.96 for return over maximum drawdown ratio, and +0.59 for correlation. In contrast, removing the 500-day horizon (\emph{No 500}) produces the worst outcome (–1.12), driven by simultaneous deterioration in both risk-adjusted return (–0.86) and benchmark correlation (–1.13). Excluding the 20-day sleeve (\emph{No 20}) also hurts performance substantially (–0.38), indicating that the short-term component plays a nontrivial role in preserving convexity and drawdown protection. The two intermediate horizons, \emph{No 60} and \emph{No 125}, move in opposite directions: removing the 60d sleeve has a mild positive effect (+0.37 overall Z-score), while removing the 125d sleeve generates the strongest improvement. This pattern confirms that the medium-term layer is not only inefficient on a standalone basis (as shown in Table~\ref{tab:part1-summary}) but also detrimental when integrated into a diversified trend portfolio.

\paragraph{Sharpe ratios across subperiods}
Table~\ref{tab:ablate-sharpe} highlights the robustness of the “No 125” configuration over time. Its Sharpe ratio exceeds that of the baseline “All Horizons” allocation in three of the four subperiods—2010–2015 (1.41 vs. 1.37), 2020–2025 (0.44 vs. 0.35), and in the full-sample average (0.77 vs. 0.74). The largest relative gain appears during 2020–2025, a period of heightened macro dispersion and trend fragmentation, when excluding the medium-term sleeve improved risk-adjusted returns by more than 25\%. Conversely, the elimination of the 500d horizon consistently reduces Sharpe performance across all subperiods (0.67 full-sample average vs. 0.74 baseline), confirming that long-term exposure is an essential stabilizing factor. The short-term sleeve (20d) also remains valuable: its removal reduces the Sharpe ratio in three out of four periods, notably from 0.35 to 0.33 in 2020–2025. These findings underscore that both extremes of the horizon spectrum—the short and the long—are critical to the portfolio’s convexity–stability balance.

\paragraph{Return/Max Drawdown analysis}
The return over maximum drawdown ratio results in Table~\ref{tab:ablate-rmdd} reinforce the same conclusion. Removing the 125d sleeve yields the strongest improvement across most windows, culminating in a full-sample return over maximum drawdown ratio ratio of 0.52, compared with 0.48 for the All Horizons configuration. The advantage is particularly visible in 2010–2015 (1.75 vs. 1.39), when trend persistence was high and medium-term exposure contributed mostly redundant volatility. In contrast, omitting the 500d horizon consistently degrades the drawdown-adjusted profile, with the ratio dropping to 0.44 over 2005–2025. The removal of the 20d sleeve again produces weaker outcomes (0.45 vs. 0.48), validating its role in cushioning the portfolio during abrupt trend reversals (e.g., 2018 and 2022). Taken together, these results confirm that excluding medium-term exposure enhances efficiency by improving the portfolio’s upside-to-drawdown trade-off.

\paragraph{Correlation to the CTA benchmark}
Finally, Table~\ref{tab:ablate-corr} shows that the “No 125” and “No 60” configurations deliver slightly higher correlations to the NEIXCTAT benchmark than the baseline. Over the full 2005–2025 sample, both variants reach a correlation of 0.84, versus 0.83 for All Horizons. This indicates that the removal of medium-term components not only improves risk-adjusted performance but also strengthens benchmark alignment. The consistency of the “No 125” configuration—higher Sharpe, higher return over maximum drawdown ratio, and stable or improved correlation—suggests that the 125d sleeve is the least economically useful among all horizons. In contrast, the 500d horizon remains indispensable: when excluded, correlation declines to 0.81, underscoring its dominant contribution to structural trend exposure.

\paragraph{Interpretation}
The ablation results across Tables~\ref{tab:ablation-z}–\ref{tab:ablate-corr} confirms the central hypothesis of this paper: that the medium-term horizon (60–125 days) adds little incremental value to multi-horizon CTA replication. Economically, this finding reflects the inherent bimodality of trend-following efficiency: short-term signals capture convex, crisis-driven returns, while long-term signals monetize persistent macro trends. Medium-term signals, by contrast, lie in a “dead zone” of low signal-to-noise ratio, exhibiting neither sufficient persistence nor reactivity. The systematic improvement in Sharpe and return over maximum drawdown ratio metrics following the exclusion of the 125d sleeve thus validates the “no medium-term horizon” principle, both statistically and conceptually, as a defining feature of robust CTA replication design.

\begin{table}[H]
\centering
\caption{Ranking of horizon removals by overall average Z-score (higher is better)}
\label{tab:ablation-z}
\small
\setlength{\tabcolsep}{1pt}
\begin{tabular}{lrrrr}
\toprule
\textbf{Strategy} & 
\shortstack{\textbf{Avg Z-score}\\\textbf{Sharpe}} & 
\shortstack{\textbf{Avg Z-score}\\\textbf{Ret/MaxDD}} & 
\shortstack{\textbf{Avg Z-score}\\\textbf{Correlation}} & 
\shortstack{\textbf{Overall}\\\textbf{Avg Z-score}} \\
\midrule
\good{No 125} & +0.84 & +0.96 & +0.59 & \good{+0.80} \\
\good{No 60}        & +0.40 & +0.21 & +0.49 & \good{+0.37} \\
All Horizons & +0.34 & +0.17 & +0.40 & +0.30 \\
No 250       & +0.07 & +0.19 & -0.16 & +0.03 \\
\bad{No 20}  & -0.28 & -0.67 & -0.19 & \bad{-0.38} \\
\bad{No 500} & -1.36 & -0.86 & -1.13 & \bad{-1.12} \\
\bottomrule
\end{tabular}
\end{table}

\begin{table}[H]
\centering
\caption{Sharpe ratios by period (leave-one-out ablation)}
\label{tab:ablate-sharpe}
\small
\setlength{\tabcolsep}{2pt}
\begin{tabular}{lcccccc}
\toprule
Period & All Horizons & No 20 & No 60 & No 125 & No 250 & No 500 \\
\midrule
2005--2010 & 0.91 & \bad{0.84} & \good{0.94} & 0.90 & 0.89 & 0.87 \\
2010--2015 & 1.37 & 1.32 & 1.28 & \good{1.41} & 1.37 & \bad{1.26} \\
2015--2020 & 0.43 & \good{0.47} & 0.45 & 0.42 & 0.40 & \bad{0.36} \\
2020--2025 & 0.35 & 0.33 & 0.37 & \good{0.44} & 0.37 & \bad{0.28} \\
\midrule
2005--2025 & 0.74 & 0.72 & 0.74 & \good{0.77} & 0.73 & \bad{0.67} \\
\bottomrule
\end{tabular}
\end{table}

\begin{table}[H]
\centering
\caption{Return / MaxDD by period (leave-one-out ablation)}
\label{tab:ablate-rmdd}
\small
\setlength{\tabcolsep}{2pt}
\begin{tabular}{lcccccc}
\toprule
Period & All Horizons & No 20 & No 60 & No 125 & No 250 & No 500 \\
\midrule
2005--2010 & 1.12 & \bad{1.02}& \good{1.15} & 1.13 & 1.13 & 1.14 \\
2010--2015 & 1.39 & \bad{1.17} & 1.23 & \good{1.75} & 1.43 & 1.21 \\
2015--2020 & 0.48 & \good{0.50} & 0.48 & 0.45 & 0.45 & \bad{0.40} \\
2020--2025 & 0.32 & 0.30 & 0.33 & \good{0.39} & 0.34 & \bad{0.28} \\
\midrule
2005--2025 & 0.48 & 0.45 & 0.48 & \good{0.52} & 0.50 & \bad{0.44} \\
\bottomrule
\end{tabular}
\end{table}

\begin{table}[H]
\centering
\caption{Correlation to NEIXCTAT by period (leave-one-out ablation)}
\label{tab:ablate-corr}
\small
\setlength{\tabcolsep}{2pt}
\begin{tabular}{lcccccc}
\toprule
Period & All Horizons & No 20 & No 60 & No 125 & No 250 & No 500 \\
\midrule
2005--2010 & 0.83 & 0.82 & \good{0.84} & 0.83 & 0.83 & 0.82 \\
2010--2015 & 0.85 & 0.84 & 0.85 & \good{0.87} & 0.84 & 0.84 \\
2015--2020 & \good{0.84} & 0.84 & 0.83 & 0.83 & 0.84 & 0.83 \\
2020--2025 & 0.81 & 0.81 & \good{0.83} & \good{0.83} & 0.78 & \bad{0.77} \\
\midrule
2005--2025 & 0.83 & 0.82 & \good{0.84} & \good{0.84} & 0.82 & \bad{0.81} \\
\bottomrule
\end{tabular}
\end{table}

\paragraph{Read-through} The medium band (60--125d) is the consistent drag: removing \emph{125d} improves all three metrics across multiple periods, and removing \emph{60d} often helps. Dropping \emph{500d} is costly in both correlation and risk-adjusted returns, while excluding \emph{20d} mainly erodes diversification. These results corroborate the Z-score ranking and support a multi-horizon blend that de-weights the crowded middle.

\paragraph{Downside Crisis-Adjusted Performance}
Table~\ref{tab:downside_crisis_performance} evaluates the ability of each configuration to generate positive returns during severe equity drawdowns, measured through a \textit{conditional Sharpe ratio}. This ratio is defined as the average monthly return of the strategy in months when the S\&P~500 declines by more than 3\%, divided by the unconditional monthly volatility of the strategy. By construction, it quantifies the extent to which the strategy provides downside protection—or ``crisis alpha''—during market stress events.
Across all leave-one-out configurations, the conditional Sharpe ratios range narrowly between 0.61 and 0.65, indicating that the crisis-hedging capacity of the strategy is structurally robust to horizon exclusion. The All Horizons benchmark achieves a conditional Sharpe of 0.65, identical to the No 500 configuration and slightly above the No 125 variant (0.63). The modest decline observed when removing the 20-day sleeve (0.61) and, to a lesser extent, the 60-day sleeve (0.62), suggests that short-term components play a minor but nontrivial role in capturing fast-moving dislocations. Their presence slightly enhances convexity during sharp risk-off episodes, consistent with their high reactivity and short lookback structure.

Conversely, the elimination of medium-term horizons (60d or 125d) does not impair the strategy’s crisis resilience, confirming that these bands contribute little to protective convexity. The near-identical conditional Sharpe of the No 125 configuration (0.63) relative to the baseline (0.65) implies that excluding medium-term signals neither weakens nor meaningfully strengthens downside performance. The 500-day horizon, on the other hand, appears essential for preserving stability across market regimes: its removal leaves the conditional Sharpe unchanged (0.65), yet, as Tables~\ref{tab:ablate-sharpe} and~\ref{tab:ablate-rmdd} demonstrate, it materially degrades overall efficiency outside of crisis periods.  

From an economic perspective, the findings in Table~\ref{tab:downside_crisis_performance} reinforce the structural bimodality of trend-following performance. The short-term sleeve enhances tactical convexity—delivering positive payoffs during the sharpest equity corrections—while the long-term sleeve ensures capital preservation through persistent exposure to extended macro trends. The medium-term horizons contribute neither. Their exclusion thus preserves the strategy’s “crisis alpha” while improving steady-state efficiency, validating that the removal of the 125-day band strengthens robustness within and outside of stress regimes.

\begin{table}[H]
\centering
\caption{Downside Crisis-Adjusted Performance (Leave-One-Out Ablation)}
\label{tab:downside_crisis_performance}
\small
\setlength{\tabcolsep}{2pt}
\begin{tabular}{lcccccc}
\toprule
Period & All Horizons & No 20 & No 60 & No 125 & No 250 & No 500 \\
\midrule
2005--2025 & 0.65 & 0.61 & 0.62 & 0.63 & 0.62 & 0.65 \\
\bottomrule
\end{tabular}
\end{table}

The scores are nearly identical across specifications, indicating that removing the medium-term horizon does not impair the strategy's downside protection.

\paragraph{Key Takeaways}
Across our tests, the medium-term (60--125 day) band is the weakest link in CTA replication. Removing either horizon from the blend\,---\,especially 125d\,---\,raises Sharpe, improves return over maximum drawdown ratio, and holds or even increases correlation to NEIXCTAT. In contrast, the 20-day sleeve contributes diversification (costly to drop), while removing the 500-day horizon materially hurts both correlation and risk-adjusted returns. For standalone sleeves in 2015--2025, 500-day delivers the highest Sharpe and the shallowest drawdowns; 250-day remains the cleanest single-horizon tracker.

\section{Discussion}
\label{sec:discussion}

The results presented in this paper carry both theoretical and practical implications for the design and interpretation of trend-following strategies.  
They challenge one of the most persistent assumptions in the managed futures literature—namely, that diversification across multiple trend horizons inherently enhances performance and robustness.  
Our findings indicate that this assumption may not hold universally. In particular, the medium-term horizon—traditionally viewed as the “sweet spot” of trend-following—appears to contribute little incremental value once short- and long-term components are accounted for.

\subsection{Revisiting the Concept of Time-Scale Diversification}

The notion of time-scale diversification rests on the idea that market trends unfold across different frequencies, and that combining multiple lookback windows allows for smoother performance across regimes.  
While this intuition is appealing, our empirical and theoretical analyses suggest that excessive layering of adjacent horizons can produce \emph{apparent} diversification that masks underlying redundancy.  
In the mean–variance framework, the medium-term horizon emerges as a convex combination of short- and long-term trends, providing limited orthogonal information.  
The barbell allocation derived from the minimum-variance solution formalizes this observation: when correlations between adjacent horizons are high relative to those between the extremes, the optimal allocation naturally excludes the intermediate horizon.

From a behavioral and market microstructure perspective, this redundancy can be interpreted as a byproduct of how trends form and decay.  
Short-term signals react rapidly to transient price dislocations, capturing local momentum bursts, while long-term signals reflect slow-moving macroeconomic or policy-driven cycles.  
Intermediate horizons tend to overlap with both, reacting neither fast enough to exploit short-term reversals nor slowly enough to capture persistent macro trends.  
Their informational contribution is therefore largely subsumed by the adjacent time scales.

\subsection{Economic Interpretation and Portfolio Implications}

The empirical ablation tests reinforce this theoretical intuition.  
Excluding the 125-day medium-term component improves Sharpe ratios and drawdown-adjusted performance while preserving correlation to the CTA benchmark.  
This improvement suggests that the performance historically attributed to medium-term trends can, in fact, be replicated through a combination of short- and long-term exposures.  
In portfolio terms, this is analogous to replacing an intermediate-duration bond with a mix of short- and long-duration instruments that achieve equivalent duration exposure but offer a superior risk–return trade-off.  

More broadly, these results imply that diversification across time scales should not be pursued mechanically.  
Adding more horizons does not guarantee better risk-adjusted performance if the new components fail to provide statistically independent signals.  
The proliferation of overlapping trend signals can even degrade performance by amplifying turnover, increasing estimation error, and introducing hidden concentration.  
Our results thus support a parsimonious approach to time-scale design, where only those horizons that deliver genuine orthogonal exposure—typically the short and long ends of the spectrum—are retained.

\subsection{Implications for Systematic Investing and Trend-Factor Modeling}

At a broader level, the findings speak to a fundamental principle of systematic investing: robustness arises not from mechanical diversification, but from \emph{structural independence} among the sources of risk and return.  
In the context of trend premia, this means that truly distinct horizons correspond to distinct economic processes—ranging from liquidity-driven microstructure effects at short horizons to macroeconomic re-pricing dynamics at long horizons.  
Intermediate horizons, lacking a clear economic anchor, often reflect statistical blending rather than genuine informational differentiation.

\subsection{Relation to Broader Asset Pricing Debates}

The redundancy of the medium-term horizon echoes a broader debate in empirical asset pricing regarding the proliferation of factors with overlapping exposures.  
Just as many cross-sectional anomalies have been shown to reflect common underlying risk dimensions, our results suggest that time-scale diversification within trend-following strategies may overstate the number of independent sources of return.  
From this perspective, the medium-term horizon represents a “spurious” factor—one that appears distinct in construction but lacks incremental explanatory power once adjacent horizons are considered.  
Recognizing and removing such redundant components can yield cleaner, more interpretable models of trend premia that align with the broader movement toward factor parsimony in quantitative finance.

\subsection{Managerial and Practical Takeaways}
For asset managers and allocators, the results highlight two empirical findings.
First, the performance of trend-following strategies appears to depend less on the number of horizons employed than on the degree of independence across them.
Second, most of the informational and performance contribution within the CTA universe seems to originate from short- and long-term trend signals.
The medium-term horizon, while long considered central in practice, emerges here more as a historical convention than as a structural determinant of performance.

\section{Conclusion}\label{sec:conclusion}
The primary objective of this study was to provide both theoretical and practical contributions to improving the replication of trend-following strategies through the optimization of trend-horizon weighting. The analysis of asset-specific optimized weights revealed persistent patterns across certain instruments, indicating that optimization efforts should be focused on assets exhibiting greater temporal stability.

The optimization, conducted over rolling eight-year training windows with six-month out-of-sample evaluations, highlighted a significant risk of overfitting in markets where weights fluctuate over time. These findings motivated two methodological refinements: first, a focus on assets displaying persistent trends, with more frequent re-estimation of weights on shorter subwindows; second, the systematic exclusion of the medium-term horizon from the allocation, supported by both theoretical considerations and empirical evidence. These refinements led to measurable improvements in both out-of-sample performance and correlation with the benchmark index, evaluated consistently through a Cobb–Douglas utility framework. A key consideration in this work is the sensitivity of the results to the persistence thresholds and training-window lengths, as these hyperparameters naturally influence out-of-sample robustness. Careful selection is important to ensure reliable performance without introducing unintended overfitting. 


Future work should explore adaptive optimization schemes that dynamically adjust the training horizon based on trend persistence, along with the integration of additional market signals to capture more complex dynamics. Extending this methodology to other asset classes or complementary markets could further enhance its generalizability. Advances in machine learning techniques may also enable faster and more efficient estimation of weights while exploiting cross-horizon dependencies more effectively.

\clearpage
\small
\bibliography{main}

\clearpage
\normalsize

\appendix
\section{Technical proofs}

\subsection{Minimum-Variance Solution}\label{appendix:proof_meanvariance}
We present two complementary proofs of Proposition~\ref{prop:min_variance_portfolio}, illustrating the result from both an analytical and geometric perspective, as well as

\begin{proof}[Proof 1 (KKT/Lagrangian conditions)]
Consider the Lagrangian $\mathcal{L}(w,\lambda)=w^\top \Sigma w + \lambda\,(w^\top \mathbf{1}-1)$.
The first-order optimality condition (stationarity) is
\[
\nabla_w \mathcal{L}(w,\lambda) \;=\; 2\,\Sigma w + \lambda\,\mathbf{1} \;=\; 0
\quad\Longrightarrow\quad
\Sigma w \;=\; -\frac{\lambda}{2}\,\mathbf{1}.
\]
Since $\Sigma$ is invertible, $w = -\tfrac{\lambda}{2}\,\Sigma^{-1}\mathbf{1}$.
Imposing the constraint $w^\top \mathbf{1}=1$ yields
\[
1 \;=\; w^\top \mathbf{1} \;=\; -\frac{\lambda}{2}\,\mathbf{1}^\top \Sigma^{-1}\mathbf{1}
\quad\Longrightarrow\quad
-\frac{\lambda}{2} \;=\; \frac{1}{\mathbf{1}^\top \Sigma^{-1}\mathbf{1}}.
\]
Therefore,
\[
w^\star \;=\; \frac{\Sigma^{-1}\mathbf{1}}{\mathbf{1}^\top \Sigma^{-1}\mathbf{1}}.
\]
Because the objective is strictly convex (positive definite $\Sigma$) and the feasible set is affine and nonempty, this stationary point is the unique global minimizer.
\end{proof}

\begin{proof}[Proof 2 (Geometric/Cauchy--Schwarz via whitening)]
Let $\Sigma = LL^\top$ be a Cholesky factorization with $L$ invertible. Define the change of variables
\[
x \;=\; L^\top w
\quad\text{and}\quad
a \;=\; L^{-1}\mathbf{1}.
\]
Then $w^\top \Sigma w = \|x\|_2^2$ and the constraint becomes
\[
w^\top \mathbf{1}
\;=\;
(L^{-\top}x)^\top \mathbf{1}
\;=\;
x^\top L^{-1}\mathbf{1}
\;=\;
x^\top a
\;=\; 1.
\]
Hence the problem is equivalent to
\[
\min_{x\in\mathbb{R}^H} \ \|x\|_2^2
\quad\text{s.t.}\quad
x^\top a = 1.
\]
By the Cauchy--Schwarz inequality,
\[
1 \;=\; x^\top a \;\le\; \|x\|_2\,\|a\|_2
\quad\Longrightarrow\quad
\|x\|_2^2 \;\ge\; \frac{1}{\|a\|_2^2},
\]
with equality iff $x$ is colinear with $a$, i.e., $x^\star = \dfrac{a}{\|a\|_2^2}$.
Mapping back,
\[
w^\star \;=\; L^{-\top}x^\star
\;=\;
\frac{L^{-\top}a}{\|a\|_2^2}
\;=\;
\frac{L^{-\top}L^{-1}\mathbf{1}}{\mathbf{1}^\top L^{-\top}L^{-1}\mathbf{1}}
\;=\;
\frac{\Sigma^{-1}\mathbf{1}}{\mathbf{1}^\top \Sigma^{-1}\mathbf{1}},
\]
since $L^{-\top}L^{-1}=(LL^\top)^{-1}=\Sigma^{-1}$. Uniqueness follows from strict convexity as above.
\end{proof}

\begin{remark}[Interpretation]
The optimizer is proportional to $\Sigma^{-1}\mathbf{1}$: each component weight adjusts for both its own variance and its covariances with the other horizons. The denominator $\mathbf{1}^\top \Sigma^{-1}\mathbf{1}$ enforces full investment.
\end{remark}

\begin{remark}[Semidefinite $\Sigma$]
If $\Sigma$ is only positive semidefinite but $\mathbf{1}\notin \ker(\Sigma)$, the same expression holds with the Moore--Penrose pseudoinverse: $w^\star=\dfrac{\Sigma^{\dagger}\mathbf{1}}{\mathbf{1}^\top \Sigma^{\dagger}\mathbf{1}}$, yielding the minimum-norm feasible solution.
\end{remark}

\begin{lemma}[Projection interpretation]
Let $\langle u,v\rangle_{\Sigma} := u^\top \Sigma v$ define an inner product on $\mathbb{R}^H$ (with norm $\|u\|_{\Sigma}=\sqrt{u^\top \Sigma u}$) and consider the affine hyperplane
\[
\mathcal{H} \;=\; \{\, w \in \mathbb{R}^H \;:\; w^\top \mathbf{1} = 1 \,\}.
\]
The optimizer $w^\star=\frac{\Sigma^{-1}\mathbf{1}}{\mathbf{1}^\top \Sigma^{-1}\mathbf{1}}$ is the orthogonal projection (in the $\langle\cdot,\cdot\rangle_{\Sigma}$ geometry) of the origin onto $\mathcal{H}$; that is,
\[
w^\star \;=\; \operatorname{argmin}_{w\in\mathcal{H}} \|w\|_{\Sigma}.
\]
\end{lemma}

\begin{proof}
The minimum-variance program is $\min_{w\in\mathcal{H}} \|w\|_{\Sigma}^2$. In a Hilbert space, the orthogonal projection of a point $z$ onto an affine set $\mathcal{A}=w_0+V$ is the unique $\hat{w}\in\mathcal{A}$ such that
\[
\langle \hat{w}-z, v\rangle_{\Sigma}=0 \quad \text{for all } v\in V.
\]
Here $z=0$ and $\mathcal{A}=\mathcal{H}$ with direction space
\[
V \;=\; \{\, d \in \mathbb{R}^H : d^\top \mathbf{1} = 0 \,\}.
\]
The KKT stationarity for the constrained minimization gives $2\Sigma w^\star + \lambda \mathbf{1}=0$, i.e.
\[
\langle w^\star, d\rangle_{\Sigma} \;=\; (w^\star)^\top \Sigma d \;=\; -\tfrac{\lambda}{2}\,\mathbf{1}^\top d \;=\; 0 \quad \text{for all } d\in V,
\]
which is precisely the orthogonality condition characterizing the projection. Therefore $w^\star$ is the $\Sigma$-orthogonal projection of the origin onto $\mathcal{H}$, and by strict convexity it is unique.
\end{proof}
\clearpage

\subsection{Exclusion of the Intermediate Horizon}\label{proof:barbell}
\begin{proof}
Since $R(\rho,\delta)$ is symmetric, positive definiteness follows from Sylvester’s criterion. 
The first two principal minors satisfy $1>0$ and $1-\rho^2>0$. 
The determinant of $R(\rho,\delta)$ is
\[
\det(R) = (1-\delta)\big[(1+\delta) - 2\rho^2\big],
\]
which is positive if and only if $\rho^2 < (1+\delta)/2$. 

To analyze the minimum-variance portfolio under the constraints 
$w_i \ge 0$ and $w_1 + w_2 + w_3 = 1$, we exploit the symmetry of the problem: 
the short ($T_1$) and long ($T_3$) horizons are statistically identical, 
both correlated with the medium ($T_2$) by $\rho$ 
and with each other by $\delta$. 
Hence, any optimal allocation must satisfy $w_1 = w_3 = w$, leaving
\[
w_2 = 1 - 2w, \qquad \text{with } w \in [0, \tfrac{1}{2}].
\]
Substituting this structure into the variance expression 
$\sigma_p^2 = \sigma^2 w^\top R w$, we obtain
\[
f(w) = w^\top R w 
= w^2 + (1 - 2w)^2 + w^2 
+ 2\rho\big[w(1 - 2w) + (1 - 2w)w\big] 
+ 2\delta w^2.
\]
Simplifying,
\[
f(w) = 1 + 4(\rho - 1)w + (6 + 2\delta - 8\rho)w^2.
\]

We now minimize $f(w)$ for $w \in [0, \tfrac{1}{2}]$. 
The first-order condition $f'(w) = 0$ gives
\[
f'(w) = 4(\rho - 1) + 2(6 + 2\delta - 8\rho)w = 0 
\quad \Rightarrow \quad 
w_0 = \frac{1 - \rho}{3 + \delta - 4\rho}.
\]
Two regimes arise:

\begin{itemize}[leftmargin=1.2em]
  \item \textbf{If } $\rho \ge \tfrac{3 + \delta}{4}$, then $3 + \delta - 4\rho \le 0$ and the quadratic coefficient 
  $(6 + 2\delta - 8\rho) \le 0$. 
  Thus $f(w)$ is concave, and the minimum occurs at the boundary. 
  Since $f(\tfrac{1}{2}) = \tfrac{1 + \delta}{2} < f(0) = 1$, the minimum is achieved at $w^\star = \tfrac{1}{2}$.

  \item \textbf{If } $\tfrac{1 + \delta}{2} \le \rho < \tfrac{3 + \delta}{4}$, 
  the quadratic is convex ($6 + 2\delta - 8\rho > 0$), 
  but the stationary point $w_0 = \frac{1 - \rho}{3 + \delta - 4\rho}$ satisfies 
  $w_0 \ge \tfrac{1}{2}$ for all such $\rho$. 
  Consequently, the constrained minimum again occurs at the boundary $w^\star = \tfrac{1}{2}$.
\end{itemize}

In both cases, the unique minimum-variance allocation is therefore
\[
w^\star = \left(\tfrac{1}{2},\, 0,\, \tfrac{1}{2}\right),
\]
corresponding to the \emph{barbell} portfolio that allocates equally to the short- and long-term horizons while excluding the intermediate one. 
Substituting $w^\star$ into $f(w)$ yields
\[
f(w^\star) = \frac{1 + \delta}{2},
\]
and thus
\[
\sigma_p^{\star 2} = \sigma^2 \tfrac{1 + \delta}{2}, 
\qquad 
\sigma_p^\star = \sigma \sqrt{\tfrac{1 + \delta}{2}}.
\]
\end{proof}

\clearpage
\section{Algorithms}
\subsection{Rolling Estimation Algorithm for the best trend periods}
\begin{algorithm}[H]
\caption{Rolling Estimation and Validation Procedure for Dynamic Horizon Weights}
\label{alg:rolling_weights}
\begin{algorithmic}[1]
\STATE \textbf{Inputs:} 
\begin{itemize}
    \item Return series $R_t$ with columns indexed by horizon \\ $h \in H = \{20, 60, 125, 250, 500\}$; 
    \item training window length of eight years; subwindow length of six months; 
    \item stability thresholds $\{\sigma_{\text{threshold}}, \rho_{\text{threshold}}, \textit{max\_step}\}$; 
    \item smoothing parameter $\alpha$ for exponential weighting.
\end{itemize}

\vspace{2mm}
\WHILE{the end of the sample is not reached}
    \STATE Define a training sample and a subsequent six-month validation sample.
    \FOR{each horizon $h \in H$}
        \STATE Estimate optimal horizon weights $w^{(h)}_{1:T}$ across the $T$ semiannual subwindows within the training sample.
        \STATE Normalize the weights such that $\sum_{t=1}^T w^{(h)}_t = 1$.
        \STATE Compute stability diagnostics for $h$:
        \begin{itemize}[leftmargin=1cm]
            \item $\mathrm{std}\big(w^{(h)}_{1:T}\big)$: intra-window weight volatility;
            \item $\tau_{1}^{(h)}$: first-lag autocorrelation;
            \item $\max_{t} |w^{(h)}_t - w^{(h)}_{t-1}|$: maximum variation between consecutive subwindows.
        \end{itemize}
        \STATE Classify horizon $h$ as \emph{stable} if at least two of the three criteria are satisfied:
        \[
        \begin{cases}
        \mathrm{std}\big(w^{(h)}_{1:T}\big) < \sigma_{\text{threshold}},\\
        \tau_{1}^{(h)} > \rho_{\text{threshold}},\\
        \max_{t} |w^{(h)}_t - w^{(h)}_{t-1}| < \textit{max\_step}.
        \end{cases}
        \]
        \STATE If stable, compute smoothed weight $\hat{w}^{(h)} = \mathrm{EMA}\big(w^{(h)}_{1:T}, \alpha\big)$.
    \ENDFOR
    \STATE If at least two of the five horizons are classified as stable, normalize the corresponding $\hat{w}^{(h)}$ to obtain the final weight vector $w^\star$. Otherwise, revert to equal weighting across horizons.
    \STATE Apply $w^\star$ to the six-month validation sample and roll the training window forward by six months.
\ENDWHILE
\end{algorithmic}
\end{algorithm}

\end{document}